\newif\iflong
\newif\ifshort

\longtrue

\iflong
\else
\shorttrue
\fi
\documentclass[pagebackref]{article}

\usepackage[utf8]{inputenc}

\usepackage{amsmath,amssymb,xspace}
\usepackage{paralist}
\iflong{}
  \usepackage{authblk,amsthm}
  \usepackage[margin=4.25cm]{geometry}
\fi{}
\usepackage{booktabs}
\usepackage{tabularx}
\usepackage{multirow}

\makeatletter
\def\NAT@spacechar{~}%
\makeatother
\usepackage{tikz}
\usetikzlibrary{decorations.pathmorphing}
\usetikzlibrary{decorations.pathreplacing}
\usetikzlibrary{shapes,automata}
\usetikzlibrary{arrows}

\pgfdeclarelayer{background}
\pgfsetlayers{background,main}

\usepackage{color}
\usepackage{graphicx}
\definecolor{mygrey}{rgb}{0.9,0.9,0.9}
\definecolor{darkgreen}{RGB}{0,100,0}

\usepackage[pagebackref,pdfdisplaydoctitle,menucolor=orange!40!black,filecolor=magenta!40!black,urlcolor=blue!40!black,linkcolor=red!40!black,citecolor=green!40!black,colorlinks]{hyperref}
\hypersetup{pdftitle={Routing with Few Collision}, pdfauthor={Till Fluschnik, Marco Morik, and Manuel Sorge}}

\usepackage[sort&compress,nameinlink,noabbrev,capitalize]{cleveref}
\usepackage[T1]{fontenc}

\usepackage{etoolbox}

\iflong{}
\theoremstyle{plain}
\newtheorem{theorem}{Theorem}
\newtheorem{corollary}{Corollary}

\newtheorem{lemma}{Lemma}

\newtheorem{observation}{Observation}
\fi{}
\iflong{}\theoremstyle{definition}\newtheorem{definition}{Definition}\newtheorem{constr}{Construction}\else{}
\spnewtheorem{constr}{Construction}{\bfseries}{\normalfont}\fi{}

\crefname{constr}{Construction}{Constructions}
\crefname{step}{Step}{Steps}
\crefname{observation}{Observation}{Observations}
\crefname{proposition}{Proposition}{Propositions}
\crefname{remark}{Remark}{Remarks}
\crefname{prop}{Property}{Properties}
\creflabelformat{prop}{(#2#1#3)}

\newcommand{\Oh}{O}

\newcommand{\N}{\mathbb{N}}
\newcommand{\I}{\mathcal{I}}

\newcommand{\calP}{\mathcal{P}}

\newcommand{\calF}{\mathcal{F}}
\newcommand{\calW}{\mathcal{W}}

\newcommand{\W}[1]{\ensuremath{\operatorname{W}[#1]}\xspace}

\newcommand{\NP}{\ensuremath{\operatorname{NP}}}

\DeclareMathOperator{\dist}{dist}

\DeclareMathOperator{\outdeg}{outdeg}
\DeclareMathOperator{\indeg}{indeg}

\newcommand{\decprob}[3]{
  \begin{center}
    \begin{minipage}{0.95\textwidth}
	    \noindent
	    \textsc{#1}\\
	    \setlength{\tabcolsep}{3pt}
	    \begin{tabularx}{\textwidth}{@{}lX@{}}
			    \normalsize \textbf{Input:} 		& \normalsize #2 \\
			    \normalsize \textbf{Question:} 	& \normalsize #3
		    \end{tabularx}
    \end{minipage}
  \end{center}
}

\newcommand{\mtse}{Routing with Collision Avoidance\xspace}
\newcommand{\mtseAcr}{\textsc{RCA}\xspace}

\newcommand{\smtse}{Fast~\mtse\xspace}
\newcommand{\smtseAcr}{\textsc{FRCA}\xspace}

\newcommand{\pmtseAcr}{\textsc{Path-\mtseAcr}\xspace}
\newcommand{\psmtseAcr}{\textsc{Path-\smtseAcr}\xspace}

\newcommand{\tmtseAcr}{\textsc{Trail-\mtseAcr}\xspace}
\newcommand{\tsmtseAcr}{\textsc{Trail-\smtseAcr}\xspace}

\newcommand{\wmtseAcr}{\textsc{Walk-\mtseAcr}\xspace}
\newcommand{\wsmtseAcr}{\textsc{Walk-\smtseAcr}\xspace}

\newcommand{\LD}{\smallskip\noindent($\Leftarrow$)}
\newcommand{\RD}{\smallskip\noindent($\Rightarrow$)}

\newcommand{\hcycle}{Hamiltonian cycle}
\usepackage{etoolbox}

\newcommand{\appsymb}{$\star$}
\newcommand{\appref}[1]{{\hyperref[proof:#1]{\appsymb}}}
\newcommand{\toappendix}[1]{%
  \iflong{}#1
  \else{}\gappto{\appendixProofText}{
    {#1}
   }\fi{}
}

\newcommand{\appendixproof}[2]{%
  \iflong{}#2
  \else{}\gappto{\appendixProofText}{
    \subsection{Proof of \cref{#1}}\label{proof:#1}#2
    }\fi{}
}

\newcommand{\appendixsection}[1]{%
  \ifshort{}\gappto{\appendixProofText}{
      \section{Additional Material for~\cref{#1}}
      \label{appsec:#1}
    
  }\fi{}
}

\newcommand{\varn}{\eta}
\newcommand{\Deltad}{\Delta_{\rm i/o}}

\newcommand{\share}{share\xspace}

\newcommand{\shared}{shared\xspace}
\newcommand{\sharing}{sharing\xspace}

\newcommand{\route}{route\xspace}
\newcommand{\routes}{routes\xspace}

\title{\thetitle{}\footnote{Major parts of this work done while all authors were with TU~Berlin.}%
}
\iflong{}
\newcommand{\fnsep}{$^,$}
\author[1]{Till~Fluschnik\thanks{Supported by the DFG, project DAMM (NI~369/13-2).}\fnsep}

\author[1]{Marco~Morik}

\author[1,2]{Manuel~Sorge\thanks{Supported by the DFG, project DAPA (NI~369/12-2), the People Programme (Marie Curie Actions) of the European Union's Seventh Framework Programme (FP7/2007-2013) under REA grant agreement number 631163.11 and by the Israel Science Foundation (grant no. 551145/14).}\fnsep}

\affil[1]{\small{Institut f\"ur Softwaretechnik und Theoretische Informatik, TU~Berlin, Germany, \texttt{till.fluschnik@tu-berlin.de},\texttt{marco.morik@campus.tu-berlin.de}}}
\affil[2]{\small{Ben Gurion University of the Negev, Be'er Sheva, Israel, \texttt{sorge@post.bgu.ac.il}}}
\date{\vspace{-1cm}}
\fi{}

\newcommand{\thetitle}{The Complexity of Routing with Few Collisions}
\newcommand{\thertitle}{The Complexity of Routing with Few Collisions}
\begin{document}
\ifshort{}
\title{\thetitle{}\thanks{The paper is eligible for the best student paper award. \iflong{}Major parts of this work done while all authors were with TU~Berlin.\fi}}
\titlerunning{\thertitle{}}
\author{Till Fluschnik\inst{1}\thanks{Supported by the DFG, project DAMM (NI~369/13-2).} \and Marco Morik\inst{1} \and Manuel Sorge\inst{1,2}\thanks{Supported by the DFG, project DAPA (NI~369/12-2), the People Programme (Marie Curie Actions) of the European Union's Seventh Framework Programme (FP7/2007-2013) under REA grant agreement number 631163.11 and by the Israel Science Foundation (grant no. 551145/14).}}
\authorrunning{T.~Fluschnik, M.~Morik, and M.~Sorge }
\institute{Institut f\"ur Softwaretechnik und Theoretische Informatik, TU~Berlin, Germany, \email{till.fluschnik@tu-berlin.de}, \email{marco.t.morik@campus.tu-berlin.de}\and
Ben Gurion University of the Negev, Be'er Sheva, Israel, \email{sorge@post.bgu.ac.il}}
\fi{}
\maketitle

\begin{abstract}
\looseness=-1 We study the computational complexity of routing multiple objects through a network in such a way that only few collisions occur:
Given a graph~$G$ with two distinct terminal vertices and two positive integers~$p$ and~$k$, the question is whether one can connect the terminals by at least~$p$ \routes (e.g. paths) such that at most~$k$ edges are time-wise \shared among them. %
We study three types of \routes: traverse each vertex at most once (paths), each edge at most once (trails), or no such restrictions (walks).
We prove that for paths and trails the problem is \NP-complete on undirected and directed graphs even if~$k$ is constant or the maximum vertex degree in the input graph is constant.
For walks, however, it is solvable in polynomial time on undirected graphs for arbitrary~$k$ and on directed graphs if~$k$ is constant.
We additionally study for all \route types a variant of the problem where the maximum length of a \route is restricted by some given upper bound.
We prove that this length-restricted variant has the same complexity classification with respect to paths and trails, but for walks it becomes \NP-complete on undirected graphs.
\end{abstract}

\section{Introduction}

\looseness=-1 We study the computational complexity of determining bottlenecks in
networks. Consider a network in which each link has a \iflong{}certain\fi{}
capacity. We want to send a set of objects from point~$s$ to point~$t$
in this network, each object moving at a
constant rate of one link per time step. We want to determine whether
it is possible to send our (predefined number of) objects without congestion and, if not, which links in the
network we have to replace by larger-capacity links to make it possible. 

Apart from determining bottlenecks, the above-described task arises
when securely routing very important persons~\cite{OmranSZ13}, or
packages in a network~\cite{AssadiENYZ14}%
, routing
container transporting vehicles~\cite{RT15}, and generally may give
useful insights into the structure and robustness of a network. %
A further motivation is congestion avoidance in routing fleets of
vehicles, a problem treated by recent commercial software
products (e.g. {\url{http://nunav.net/}}) and poised to become more
important as passenger cars and freight cars become more and more
connected. Assume that we have many requests on computing a route for
a set of vehicles from a source location to a target location, as it
happens in daily commuting traffic. Then the idea is to centrally
compute these routes, taking into account the positions in space and
time of all other vehicles. To avoid congestion, we try to avoid that
on two of the routes the same street appears at the same
time.

\iflong{}
A first approximation to determine such bottlenecks would be to compute
the set of minimum cuts between~$s$ and~$t$. However, by daisy
chaining our objects, we may avoid such ``bottlenecks'' and, hence,
save on costs for improving the capacity of our links. Apart from the
(static) routes we have to take into account the traversals in time
that our objects take.%
\fi{}

Formally, we are given an undirected or directed graph with
marked source and sink vertex. We ask whether we can construct \routes
between the source and the sink in such a way that these \routes
\share as few edges as possible. %
By \routes herein we mean either paths, trails, or walks, modeling
different restrictions on the \routes: A \emph{walk} is a sequence of vertices
such that for each consecutive pair of vertices in the sequence there
is an edge in the graph.  A \emph{trail} is a walk where each edge of the
graph appears at most once.  A \emph{path} is a trail that contains each
vertex at most once.  We say that an edge is %
\emph{\shared} by two
\routes, if the edge appears at the same position in the sequence
of the two \routes.  The sequence of a \route can be interpreted
as the description of where the object taking this \route is at which time.
So we arrive at the following core problem:
\decprob{\mtse (\mtseAcr{})}
{\iflong{}A graph~$G=(V,E)$, two distinct vertices~$s,t\in V$, and two integers~$p\geq 1$ and~$k\geq 0$.
\fi{}}
  {Are there $p$ $s$-$t$~\routes that \share at most $k$~edges?}

\looseness=-1  This definition is inspired by the \textsc{Minimum Shared Edges} (MSE)
  problem~\cite{FluschnikKNS15,OmranSZ13,YeLLZ13}, in which an edge is
  already shared if it occurs in two \routes, regardless of the
  time of traversal. Finally, note that finding
  \routes from~$s$ to $t$ also models the general case of finding
  \routes between a set of sources and a set of sinks.

\looseness=-1 Considering our introductory motivating scenarios, it is reasonable to restrict the maximal length of the \routes.
For instance, when routing vehicles in daily commuting traffic while avoiding congestion, the routes should be reasonably short.
Motivated by this, we study the following variant of \mtseAcr{}.
\decprob{\smtse (\smtseAcr)}
{\iflong{}A graph~$G=(V,E)$, two distinct vertices~$s,t\in V$, and three integers $p,\alpha\geq 1$ and~$k\geq 0$.
  \fi}
  {Are there $p$ $s$-$t$~\routes each of length at most~$\alpha$ that \share at most $k$~edges?}
\looseness=-1 In the problem variants \pmtseAcr{}, \tmtseAcr{}, and \wmtseAcr{},  the \routes are restricted to be paths, trails, or walks, respectively (analogously for \smtseAcr{}).
\paragraph{Our Contributions.}

\newcommand{\ei}{each}
\newcommand{\smtab}[1]{\scriptsize#1}
\renewcommand{\arraystretch}{1.2}
\begin{table}[t]
  \iflong{}
    \setlength{\tabcolsep}{1.5pt}
  \fi{}
  \centering
  \caption{Overview of our results: DAGs abbreviates directed acyclic graphs; NP-c., W[2]-h., P abbreviate NP-complete, W[2]-hard, and containment in the class~$P$, respectively; $\Delta$ denotes the maximum degree; $\Delta_{\rm i/o}$ denotes the maximum over the in- and outdegrees. $^a$~(Thm.~\ref{thm!dagsxp}) $^b$~(Thm.~\ref{thm!dagshard}) $^c$~(Cor.~\ref{cor:smtsedagshard}) $^d$~(even on planar graphs)}
  \begin{tabular}{@{}llccllllll@{}}  \toprule
		&& \multicolumn{2}{c}{Undirected, with~$k$} & &  \multicolumn{2}{c}{Directed, with~$k$} & & \multicolumn{2}{c}{DAGs, with~$k$}  \\\cmidrule{3-4} \cmidrule{6-7} \cmidrule{9-10}
		&& \multicolumn{1}{c}{constant} & \multicolumn{1}{c}{arbitrary}		 &	& \multicolumn{1}{c}{constant} & \multicolumn{1}{l}{arbitrary} 		& & \multicolumn{1}{c}{const.} & \multicolumn{1}{c}{arbitrary} \\ \midrule
   \textsc{Path-(F)RCA} 	&&	\multicolumn{2}{c}{NP-c.$^d$, Thm.~\ref{thm!pmtsehard} (Cor.~\ref{cor:psmtsehard})} &  & \multicolumn{2}{c}{NP-c.$^d$, Thm.~\ref{thm!pmtsehard} (Cor.~\ref{cor:psmtsehard})} & & P$^a$ &  NP-c., \\\vspace{2.5pt} 
    		&& 	\smtab{\ei{}~$k\geq 0$}  & \smtab{\ei{}~$\Delta\geq 4$} & & \smtab{\ei{}~$k\geq 0$} & \smtab{\ei{}~$\Delta_{\rm i/o}\geq 4$}  & & &   \W{2}-h$^{b/c}$ \\
  \textsc{Trail-(F)RCA} 	&& \multicolumn{2}{c}{NP-c.$^d$, Thm.~\ref{thm!tmtseundirhard} (Cor.~\ref{cor:tsmtsehardundir})} & & \multicolumn{2}{c}{NP-c.$^d$, Thm.~\ref{thm!tmtsedirhard} (Cor.~\ref{cor:tsmtseharddir})} & & P$^a$ & NP-c., \\\vspace{2.5pt} 
        	&& \smtab{\ei{}~$k\geq 0$} & \smtab{\ei{}~$\Delta\geq 5$} & & \smtab{\ei{}~$k\geq 0$} & \smtab{\ei{}~$\Delta_{\rm i/o}\geq 3$} & & &  \W{2}-h$^{b/c}$\\

  \wmtseAcr{} 	&& \multicolumn{2}{c}{P (Thm.~\ref{thm!wmtseundirptime})} & & P (Thm.~\ref{thm!wmtsedirxp}) & \NP-c., & & P$^a$ & NP-c., \\\vspace{2.5pt} 
		&& & & & &  \W{2}-h.$^b$ & & &  \W{2}-h.$^b$ \\
  \wsmtseAcr{}	&& {open} & \NP-c., \W{2}-h.& & P (Thm.~\ref{thm!wmtsedirxp}) & \NP-c., & & P$^a$ & NP-c., \\
		&& & (Thm.~\ref{thm:wsmtseundirhard})  & & & \W{2}-h.$^c$  & & & \W{2}-h.$^c$ \\

  \bottomrule
  \end{tabular}

  \label{results-table}
\end{table}
We give a full computational complexity classification \ifshort{}(see \Cref{results-table})\fi{} of \mtseAcr and \smtseAcr (except \wsmtseAcr) with respect to the three mentioned \route types; with respect to undirected, directed, and directed acyclic input graphs; and distinguishing between constant and arbitrary budget.
\iflong{}\Cref{results-table} summarizes our results.\fi

To our surprise, there is no difference between paths and trails in our %
classification.
Both \pmtseAcr (\cref{sec:pmtse}) and \tmtseAcr (\cref{sec:tmtse}) are \NP-complete in all of our cases except on directed acyclic graphs when~$k\geq0$ is constant (\cref{sec:all}).
We show that the problems remain \NP-complete on undirected and directed graphs even if~$k\geq 0$ is constant or the maximum degree is constant.
We note that the \textsc{Minimum Shared Edges} problem is solvable in polynomial time when the number of shared edges is constant, 
highlighting the difference to its time-variant \pmtseAcr{}.

The computational complexity of the length-restricted variant~\smtseAcr for paths and trails equals the one of the variant without length restrictions. The variant concerning walks (\cref{sec:wmtse}) however differs from the other two variants as it is tractable in more cases, in particular on undirected graphs.
(We note that almost all of our tractability results rely on flow computations in time-expanded networks (see, e.g., Skutella~\cite{Skutella2009}).)
 Remarkably, the tractability does not transfer to the length-restricted variant~\wsmtseAcr{}, as it becomes \NP-complete on undirected graphs.
This is the only case where \mtseAcr and \smtseAcr differ with respect to their computational complexity.

\paragraph{Related Work.}

\looseness=-1 As mentioned, \textsc{Minimum Shared Edges} inspired the
definition of \mtseAcr. MSE is NP-hard on directed~\cite{OmranSZ13}
and undirected~\cite{Flu15,FluschnikKNS15} %
graphs. In contrast to \mtseAcr, if the
number of shared edges equals zero, then MSE is solvable in polynomial
time.  Moreover, MSE is \W{2}-hard with respect to the number of shared
edges and fixed-parameter tractable with respect to the number of
paths~\cite{FluschnikKNS15}.  MSE is polynomial-time solvable on
graphs of bounded treewidth~\cite{YeLLZ13,AokiHHIKZ14}.

\looseness=-1 There are various tractability and hardness results for problems related to \mtseAcr\ with $k = 0$ in temporal graphs, in which
edges are only available at predefined time steps~\cite{Berman96,KempeKK02,Michail16,MertziosMCS13}. The goal herein is
to find a number of edge or vertex-disjoint time-respecting paths connecting two
fixed terminal vertices. Time-respecting means that the time
steps of the edges in the paths are nondecreasing. Apart from the fact that all graphs that we study are 
static, the crucial difference is in the type of routes: vehicles moving along time-respecting paths may wait an arbitrary number of
time steps at each vertex, while we require them to move at least one
edge per time step (unless they already arrived at the target vertex).

Our work is related to flows over time, a concept already introduced
by Ford and Fulkerson~\cite{FF62} to measure the maximal throughput in
a network over a fixed time period. This and similar problems were
studied continually, see Skutella~\cite{Skutella2009} and Köhler et
al.~\cite{KohlerMS09} for surveys. In contrast, our throughput is
fixed, our flow may not stand still or go in circles arbitrarily, and
we want to augment the network to allow for our throughput.

\section{Preliminaries}
\label{sec:prelims}
\appendixsection{sec:prelims}

\ifshort{}
  We use basic notation from parameterized complexity~\cite{DowneyF99}.
\fi{}

We define~$[n]:=\{1,\ldots,n\}$ for every~$n\in\N$.
Let $G=(V,E)$ be an undirected (directed) graph.
Let the sequence $P=(v_1,\ldots,v_\ell)$ of vertices in~$G$ be a walk, trail, or path.
We call $v_1$ and $v_\ell$ the start and end of $P$.
For $i\in[\ell]$, we denote by~$P[i]$ the vertex~$v_i$ at position~$i$ in~$P$.
Moreover, for $i,j\in [\ell]$, $i<j$, we denote by $P[i,j]$ the subsequence $(v_i,\ldots,v_j)$ of~$P$.
By definition, $P$ has an alternative representation as sequence of edges (arcs) $P=(e_1,\ldots,e_{\ell-1})$ with $e_i:=\{v_i,v_{i+1}\}$ ($e_i:=(v_i,v_{i+1})$) for $i\in [\ell-1]$.
Along this representation, we say that $P$ \emph{contains/uses edge (arc)~$e$ at time step~$i$} if edge (arc)~$e$ appears at the $i$th position in~$P$ represented as sequence of edges (arcs) (analogue for vertices).
We call an edge/arc \emph{\shared} if two routes uses the edge/arc at the same time step. 
We say that a walk/trail/path $Q$ is an $s$-$t$~walk/trail/path, if $s$ is the start and $t$ is the end of~$Q$.
The \emph{length} of a walk/trail/path is the number of edges (arcs) contained, where we also count multiple occurrences of an edge (arc) (we refer to a path of length~$m$ as an \emph{$m$-chain}).
(We define the maximum over in- and outdegrees in~$G$ by~$\Deltad(G):=\max_{v\in V(G)}\{\outdeg(v),\indeg(v)\}$.)

\iflong{}
A parameterized problem~$P$ is a set of tuples $(x,\ell)\in \Sigma^*\times \N$, where~$\Sigma$ denotes a finite alphabet.
A parameterized problem~$P$ is \emph{fixed-parameter tractable} if it admits an algorithm that decides every input~$(x,\ell)$ in~$f(\ell)\cdot |x|^{\Oh(1)}$ time (FPT-time), where~$f$ is a computable function. 
The class FPT is the class of fixed-parameter tractable problems.
The classes \W{q}, $q\geq 1$, contain parameterized problems that are presumably not fixed-parameter tractable.
For two parameterized problems~$P$ and~$P'$, a parameterized reduction from~$P$ to~$P'$ is an algorithm that maps each input $(x,
\ell)$ to~$(x',\ell')$ in FPT-time such that~$(x,\ell)\in P$ if and only if~$(x',\ell')\in P'$, and~$\ell'\leq g(\ell)$ for some function~$g$.  
A parameterized problem~$P$ is \W{q}-hard if for every problem contained in \W{q} there is a parameterized reduction to~$P$.   
\fi{}

\ifshort{}
\looseness=-1  We state some preliminary observations on \mtseAcr{} and \smtseAcr{}.
  If the terminals~$s$ and~$t$ have distance at most~$k$, 
  then routing any number of paths along the shortest path between them introduces at most~$k$ \shared edges.
  Moreover, one can show that \mtseAcr{} and \smtseAcr{} are contained in~\NP.
  Due to space constraints we defer the details of these and other results (marked by \appsymb) to the appendix.%
\fi{}
\toappendix
{
  \paragraph{Preliminary Observations on RCA and FRCA.}
  We state some preliminary observations on \mtseAcr{} and \smtseAcr{}.
  If there is a shortest path between the terminals~$s$ and~$t$ of length at most~$k$, then routing any number of paths along the shortest path introduces at most~$k$ \shared edges.
  Hence, we obtain the following.

  \begin{observation}
  Let $(G,s,t,p,k)$ be an instance of \mtseAcr with $\dist(s,t)<\infty$.
  If $k\geq \dist_G(s,t)$, then $(G,s,t,p,k)$ is a yes-instance.
  \end{observation}
  If we consider walks, the length of an $s$-$t$~walk in a graph can be arbitrarily large.
  We prove, however, that for paths, trails, and walks, \mtseAcr{} and \smtseAcr{} are contained in~\NP, that is, each variant allows for a certificate of size polynomial in the input size that can be verified in time polynomial in the input size.

  \begin{lemma}%
    \label{lem:continNP}
  \mtseAcr{} and \smtseAcr{} on undirected and on directed graphs are contained in~\NP.
  \end{lemma}

  {
    \begin{proof}
    Given an instance $(G,s,t,p,k)$ of \pmtseAcr{} and a set of $p$ $s$-$t$~paths, we can check in polynomial time whether they \share at most~$k$ edges.
    The same holds for~\tmtseAcr{} and \wmtseAcr{} (the latter follows from~\cref{thm!wmtseundirptime}).
    This is still true for all variants on directed graphs (for walks we refer to~\cref{lem:walkmtseShortSol}).
    Moreover, we can additionally check in linear time whether the length of each path/trail/walk is at most some given~$\alpha\in\N$.
    Hence, the statement follows.
    \ifshort{}\qed\fi{}
    \end{proof}
  }
}

\section{Everything is Equal on DAGs}
\label{sec:all}
\appendixsection{sec:all}

Note that on directed acyclic graphs, every walk contains each edge and each vertex at most once.
Hence, every walk is a path in DAGs, implying that all three types of \routes are equivalent in DAGs.

We prove that \mtseAcr{} is solvable in polynomial time if the number~$k$ of \shared arcs is constant, but \NP-complete if~$k$ is part of the input.
Moreover, we prove that the same holds for the length-restricted variant~\smtseAcr{}.
We start the section with the case of constant~$k\geq0$.

\subsection{Constant Number of Shared Arcs}

\begin{theorem}
 \label{thm!dagsxp}
 \mtseAcr{} and \smtseAcr{} on $n$-vertex $m$-arc DAGs are solvable in $\Oh(m^{k+1} \cdot n^3)$~time and $\Oh(m^{k+1} \cdot \alpha^2\cdot n)$~time, respectively.
\end{theorem}
We prove~\cref{thm!dagsxp} as follows: We first show that \mtseAcr{} and \smtseAcr{} on DAGs are solvable in polynomial time if~$k = 0$ (\cref{thm!dagsptime} below).
We then show that an instance of \mtseAcr{} and \smtseAcr{} on directed graphs is equivalent to deciding, for all $k$-sized subsets~$K$ of arcs, the instance with $k=0$ and a modified input graph in which each arc in~$K$ has been copied $p$ times:

\begin{theorem}%
  \label{thm!dagsptime}
 If $k=0$, \mtseAcr on $n$-vertex $m$-arc DAGs is solvable in $\Oh(n^3\cdot m)$~time.
\end{theorem}
We need the notion of time-expanded graphs.
\iflong\begin{definition}
  \label{def:timeexpg}\fi
  Given a directed graph $G$, we denote a directed graph $H$ the (directed) \emph{$\tau$-time-expanded graph} of $G$ if 
  $V(H)=\{v^i\mid v\in V(G), i = 0, \ldots, \tau\}$ and $A(H)=\{(v^{i-1},w^{i})\mid i\in[\tau], (v,w)\in A(G)\}$.
\iflong\end{definition}\fi
Note that for every directed $n$-vertex $m$-arc graph the $\tau$-time-expanded graph can be constructed in $\Oh(\tau\cdot(n+m))$~time.
We prove that we can decide \mtseAcr and \smtseAcr by flow computation in the time-expanded graph of the input graph:%
\begin{lemma}%
 \label{lem:flows}
 Let $G=(V,A)$ be a directed graph with two distinct vertices $s,t\in V$.
 Let $p\in\N$ and $\tau:=|V|$.
 Let $H$ be the $\tau$-time-expanded graph of~$G$ with $p$ additional arcs $(t^{i-1},t^i)$ between the copies of~$t$ for each $i\in [\tau]$.
 Then, $G$ allows for at least~$p$ $s$-$t$~walks of length at most~$\tau$ not \sharing any arc if and only if $H$ allows for an $s^0$-$t^\tau$ flow of value at least~$p$.
\end{lemma}
\appendixproof{lem:flows}
{
  \begin{proof}
  \RD{}
  Let $G$ allow for $p$ $s$-$t$~walks $W_1,\ldots,W_p$ not \sharing any arc.
  We construct an $s^0$-$t^\tau$~flow of value~$p$ in $H$ as follows.
  Observe that~$W_i=(v_0,\ldots,v_\ell)$ corresponds to a path $P=(v_0^0,v_1^1,\ldots,v_\ell^\ell)$ in~$H$.
  If $\ell<\tau$, then extend this path to the path~$P=(v_0^0,\ldots,v_\ell^\ell,t^{\ell+1},,\ldots,t^\tau)$ (observe that $v_\ell^\ell=t^\ell$).
  Set the flow on the arcs of $P$ to one.
  From the fact that $W_1,\ldots,W_p$ are not \sharing any arc in~$G$, we extend the flow as described above for each walk by one, hence obtaining an $s^0$-$t^\tau$ flow of value~$p$.
  
  \LD{}
  Let~$H$ allow an $s^0$-$t^\tau$~flow of value~$p$.
  It is well-known that any $s^0$-$t^\tau$~flow of value~$p$ in~$H$ can be turned into~$p$ arc-disjoint $s^0$-$t^\tau$~paths in~$H$~\cite{KT06}.
  Let $P=(v_0^0,v_1^1,\ldots,v_\ell^\ell)$ be an $s^0$-$t^\tau$~path in~$H$.
  Let $\ell'$ be the smallest index such that~$v_{\ell'}^{\ell'}=t^{\ell'}$.
  Then $W=(v_1,v_2,\ldots,v_{\ell'})$ is an $s$-$t$~walk in~$G$.
  Let~$\calP$ be the set of $p$~$s^0$-$t^\tau$~paths in~$H$ obtained from an $s^0$-$t^\tau$~flow of value~$p$, and let~$\calW$ be the set of $p$ $s$-$t$~walks in~$G$ obtained from~$\calP$ as described above.
  As every pair of paths in~$\calP$ is arc-disjoint, no pair of walks in~$\calW$ \share any arc in~$G$.
  \ifshort{}\qed\fi{}
  \end{proof}
}
\appendixproof{thm!dagsptime}
{
  \begin{proof}[\iflong{}Proof of~\cref{{thm!dagsptime}}\fi{}]
  Let~$(G=(V,A),s,t,p,0)$ be an instance of \wmtseAcr{} with~$G$ being a directed acyclic graph.
  Let $n:=|V|$. 
  We construct the directed $n$-time-expanded graph~$H$ of $G$ with~$p$ additional arcs $(t^{i-1},t^i)$ for each $i\in [\tau]$.
  Note that any $s$-$t$ path in~$G$ is of length at most $n-1$ due to~$G$ being directed and acyclic.
  The statement then follows from~\cref{lem:flows}.
    \ifshort{}\qed\fi{}
  \end{proof}
}
\cref{lem:flows} is directly applicable to \smtseAcr, by constructing an $\alpha$-expanded graph.%
\begin{corollary}
 \label{cor:smtsedags}
 If $k=0$, then \smtseAcr on $n$-vertex $m$-arc DAGs is solvable in $\Oh(\alpha^2\cdot n\cdot m)$~time.
\end{corollary}
Let $G=(V,A)$ be a directed graph and let $K\subseteq A$ and $x\in\N$.
We denote by~$G(K,x)$ the graph obtained from~$G$ by replacing each arc~$(v,w)\in K$ in~$G$ by $x$ copies~$(v,w)_1,\ldots,(v,w)_x$.

\begin{lemma}%
 \label{lem:reductiontozero}
 Let $(G=(V,A),s,t,p,k)$ be an instance of \wmtseAcr{} with~$G$ being a directed graph.
 Then, $(G,s,t,p,k)$ is a yes-instance of \wmtseAcr{} if and only if there exists a set~$K\subseteq A$ with $|K|\leq k$ such that $(G(K,p),s,t,p,0)$ is a yes-instance of \wmtseAcr{}.
 The same statement holds true for \wsmtseAcr{}.
\end{lemma}

\appendixproof{lem:reductiontozero}
{
  \begin{proof}
  \RD{}
  Let~$G$ allow for a set of $p$ $s$-$t$ walks $\calW=\{W_1,\ldots,W_p\}$ \sharing at most $k$ edges.
  Let $K\subseteq A$ denote the set of at most $k$ arcs \shared by the walks in~$\calW$.
  We construct a set of $p$ $s$-$t$ walks $\calW'=\{W_1',\ldots,W_p'\}$ in~$G(K,p)$ from~$\calW$ as follows.
  For each $i\in[p]$, let $W_i'=W_i$. 
  Whenever an arc $(v,w)\in K$ appears in~$W_i'$, we replace the arc by its copy~$(v,w)_i$.
  Observe that (i) $W_i'$ forms an $s$-$t$ walk in $G(K,p)$, (ii) the positions of the arcs in~$A\setminus K$ in the walks remain unchanged, and (iii) for each arc $(v,w)\in K$, no walk contains the same copy of the arc.
  As the arcs in~$K$ are the only \shared arcs of the walks in~$\calW$, the walks in~$\calW'$ do not \share any arc in~$G(K,p)$.
  
  \LD{}
  Let $K\subseteq A$ be a subset of arcs in~$G$ with $|K|\leq k$ such that $G(K,p)$ allows for a set of $p$ $s$-$t$ walks $\calW'=\{W_1',\ldots,W_p'\}$ with no \shared arc.
  We construct a set of $p$ $s$-$t$ walks $\calW=\{W_1,\ldots,W_p\}$ in~$G$ from~$\calW'$ as follows.
  For each $i\in[p]$, let $W_i=W_i'$. 
  Whenever an arc $(v,w)_x$, for some~$x\in[p]$, with $(v,w)\in K$ appears in~$W_i'$, we replace the arc by its original~$(v,w)$.
  Observe that (i) $W_i$ forms an $s$-$t$ walk in $G$, (ii) the positions of the arcs in~$A\setminus K$ in the walks remain unchanged. 
  As the arcs in the set~$K$ of at most~$|K|\leq k$ arcs can appear at the same positions in any pair of two walks in~$\calW$, the $s$-$t$ walks in~$\calP$ \share at most~$k$ arcs in~$G$.
  
  \smallskip\noindent Note that as the length of the walks do not change in the proof, the statement of the lemma also holds for~\wsmtseAcr{}.
  \ifshort{}\qed\fi{}
  \end{proof}
}%
{
  \begin{proof}[\iflong{}Proof of~\cref{thm!dagsxp}\fi{}]
  Let~$(G=(V,A),s,t,p,k)$ be an instance of \wmtseAcr{} with~$G$ being a directed acyclic graph.
  For each $k$-sized subset~$K\subseteq A$ of arcs in~$G$, we decide the instance $(G(K,p),s,t,p,0)$.
  The statement for~\mtseAcr{} then follows from~\cref{lem:reductiontozero} and~\cref{thm!dagsptime}.
  We remark that the value of a maximum flow between two terminals in an $n$-vertex $m$-arc graph can be computed in~$\Oh(n\cdot m)$~time~\cite{Orlin13}.
  The running time of the algorithm is in~$\Oh(|A|^k\cdot (|V|^3\cdot |A|))$.
  The statement for~\smtseAcr{} follows analogously with~\cref{lem:reductiontozero} and~\cref{cor:smtsedags}.
  \ifshort{}\qed\fi{}
  \end{proof}
}
\subsection{Arbitrary Number of Shared Arcs}
If the number~$k$ of \shared arcs is \iflong{}part of the input\fi, then both \mtseAcr{} and \smtseAcr{} are \iflong\NP-complete and \W{2}-hard with respect to~$k$\fi.

\begin{theorem}%
  \label{thm!dagshard}
  \mtseAcr on DAGs is \NP-complete and \W{2}-hard with respect to~$k$.
\end{theorem}
The construction in the reduction for~\cref{thm!dagshard} is similar to the one used by Omran et al.~\cite[Theorem~2]{OmranSZ13}. %
Herein, we give a (parameterized) many-one reduction from the \NP-complete~\cite{Karp72} \textsc{Set Cover} problem: given a set $U=\{u_1,\ldots,u_n\}$, a set of subsets~$\calF=\{F_1,\ldots,F_m\}$ with $F_i\subseteq U$ for all $i\in[m]$, and an integer~$\ell\leq m$, is there a subset~$\calF'\subseteq \calF$ with $|\calF'|\leq \ell$ such that $\bigcup_{F\in \calF'} F = U$. 
\iflong{}We say that $\calF'$ is a \emph{set cover} and we say that the elements in $F \in \calF$ are \emph{covered} by~$F$. \fi{}
Note that \textsc{Set Cover} is \W{2}-complete with respect to the solution size~$\ell$ in question~\cite{DowneyF13}.
In the following \cref{constr:dagshard}, given a \textsc{Set Cover} instance, we construct the DAG in an equivalent \mtseAcr\ or \smtseAcr\ instance.
\begin{constr}
 \label{constr:dagshard}
\looseness=-1 Let a set $U=\{u_1,\ldots,u_n\}$, a set of subsets~$\calF=\{F_1,\ldots,F_m\}$ with $F_i\subseteq U$ for all $i\in[m]$, and an integer~$\ell\leq m$ be given.
 Construct a directed acyclic graph~$G=(V,A)$ as follows.
 Initially, let $G$ be the empty graph. 
 Add the vertex sets $V_U=\{v_1,\ldots,v_n\}$ and $V_\calF=\{w_1,\ldots,w_m\}$, corresponding to~$U$ and~$\calF$, respectively.
 Add the edge $(v_i,w_j)$ to $G$ if and only if $u_i\in F_j$.
 Next, add the vertex $s$ to $G$.
 For each $w\in V_\calF$, add an $(\ell+2)$-chain to $G$ connecting~$s$ with~$w$, and direct all edges in the chain from~$s$ towards~$w$.
 For each $v\in V_U$, add an $(\ell+1)$-chain to~$G$ connecting $s$ with~$v$, and direct all edges in the chain from~$s$ towards~$v$.
 Finally, add the vertex~$t$ to~$G$ and add the arcs $(w,t)$ for all $w\in V_\calF$. 
 \qed
\end{constr}

\begin{lemma}%
 \label{lemma:dagshard}
 Let $U$,~$\calF$,~$\ell$, and~$G$ as in \cref{constr:dagshard}.
 Then there are at most~$\ell$ sets in $\calF$ such that their union is~$U$ if and only if $G$ admits $n+m$ $s$-$t$~walks \sharing at most~$\ell$ arcs in~$G$.
\end{lemma}
\appendixproof{lemma:dagshard}
{
  \begin{proof}
  \RD{}
  Suppose there are $\ell$ sets $F_1',\ldots,F_\ell'\in \calF$ such that $\bigcup_{i\in[\ell]} F_i'=U$. 
  Let $w_1',\ldots,w_\ell'$ be the vertices in~$V_\calF$ corresponding to~$F_1',\ldots,F_k'$.
  We construct $n+m$ $s$-$t$ walks as follows.
  Each outgoing chain on $s$ corresponds to exactly one $s$-$t$~walk. 
  Those walks that start with the chains connecting $s$ with a $w\in V_\calF$ are extended directly to $t$ (there is no other choice).
  For all the other walks, as $F_1',\ldots,F_k'$ cover $U$, for each~$v_i\in V_U$ there is at least one arc towards $\{w_1',\ldots,w_k'\}$.
  Route the walks arbitrarily towards one out of~$\{w_1',\ldots,w_k'\}$ and then forward to~$t$.
  Observe that all walks contain exactly one vertex in $V_\calF$ at time~step~$\ell+3$.
  Moreover, only the arcs $(w_i',t)$ for $i\in[\ell]$ are contained in more than one walk.
  As they are at most~$\ell$, the claim follows.
  
  \LD{}
  Suppose $G$ admits a set~$\mathcal{W}$ of $n+m$ $s$-$t$~walks \sharing at most~$\ell$ arcs in~$G$.
  Observe first that the arcs of the form $(w,t)$, $w\in V_\calF$, are the only arcs that can be \shared whenever at most~$\ell$ arcs are \shared, due to the fact that each outgoing chain on~$s$ is of length longer than~$k$.
  Moreover, each arc~$(w,t)$, $w\in V_\calF$ is contained in at least one $s$-$t$ walk in~$\mathcal{W}$, because no two walks in $\mathcal{W}$ can share a chain outgoing from~$s$, and for each~$w\in V_\calF$ the only outgoing arc on~$w$ has endpoint~$t$.
  Denote by $W\subseteq V_\calF$ the set of vertices such that the set~$\{(w,t)\mid w\in W\}$ is exactly the set of \shared arcs by the $n+m$ $s$-$t$~walks in~$\mathcal{W}$.
  Observe that~$|W| \leq \ell$.
  We claim that the set of sets~$\calF'$, containing the sets that by construction correspond to the vertices in~$W$, forms a set cover for~$U$.
  We show that for each element~$u\in U$ there is a $w\in W$ such that the set corresponding to~$w$ is containing~$u$.
  
  Let $u\in U$ be an arbitrary element of~$U$.
  Consider the walk~$P \in \mathcal{W}$ containing the vertex~$v\in V_U$ corresponding to element~$u$.
  As $P$ forms an $s$-$t$~walk in $G$, walk~$P$~contains a vertex~$w'\in V_\calF$.
  As discussed before, there is a walk~$P'$ containing the chain from~$s$ to~$w'$ not containing any vertex in~$V_U$.
  By construction, $w'$ is at time~step~$\ell+3$ in both~$P$ and~$P'$.
  As the only outgoing arc on~$w'$ is~$(w',t)$, both~$P$ and $P'$ use the arc~$(w',t)$ at time step~$\ell+3$, and hence~$(w,t)$ is \shared by~$P$ and $P'$.
  It follows that~$w'\in W$, and hence~$u$ is covered by the set in~$\calF'$ corresponding to~$w'$.
  \ifshort{}\qed\fi{}
  \end{proof}%
}%
\ifshort{}%
\cref{thm!dagshard} follows then from \cref{constr:dagshard} and \cref{lemma:dagshard} (see~\cref{proof:thm!dagshard}).

\fi{}
\appendixproof{thm!dagshard}
{
  \begin{proof}[\iflong{}Proof of~\cref{thm!dagshard}\fi{}]
  We give a (parameterized) many-one reduction from \textsc{Set Cover} to \mtseAcr{}.
  Let $(U,\calF,\ell)$ be an instance of~\textsc{Set Cover}.
  We construct the instance $(G,s,t,p,k)$, where $G$ is obtained by applying~\cref{constr:dagshard}, $p=|U|+|\calF|$, and $k=\ell$.
  The correctness of the reduction then follows from~\cref{lemma:dagshard}.
  Finally, note that as $k=\ell$ and \textsc{Set Cover} is \W{2}-hard with respect to the size~$\ell$ of the set cover, it follows that \mtseAcr{} is \W{2}-hard with respect to the number~$k$ of \shared arcs.
  \ifshort{}\qed\fi{}
  \end{proof}
}%
Observe that each $s$-$t$ walk in the graph obtained from~\cref{constr:dagshard} is of length at most~$\ell+3$.
\iflong{}
Hence, in the proof of \cref{thm!dagshard}, we can instead reduce to an instance $(G,s,t,p,k,\alpha)$ of~\smtseAcr{}, where~$G$ is obtained by applying~\cref{constr:dagshard}, $p=|U|+|\calF|$, and $k=\ell$, and $\alpha=\ell+3$.
\fi{}
Therewith, we obtain the following.
\begin{corollary}
 \label{cor:smtsedagshard}
 \smtseAcr on DAGs is \NP-complete and \W{2}-hard with respect to~$k + \alpha$.
\end{corollary}
\section{Path-\mtseAcr}
\label{sec:pmtse}
\appendixsection{sec:pmtse}

In this section, we %
prove the following theorem.
\begin{theorem}%
  \label{thm!pmtsehard}
\pmtseAcr both on undirected planar and directed planar graphs is \NP-complete, even if $k \geq 0$ is constant or $\Delta\geq 4$ is constant.
\end{theorem}
In the proof of~\cref{thm!pmtsehard}, we reduce from the following NP-complete~\cite{GareyJT76} problem (a cubic graph is a graph where every vertex has degree exactly three):

\decprob{Planar Cubic Hamiltonian Cycle (PCHC)}
  {An undirected, planar, cubic graph $G$.}
  {Is there a cycle in $G$ that visits each vertex exactly once?}
  Roughly, the instance of \pmtseAcr obtained in the reduction consists of the original graph~$G$ connected to the terminals~$s, t$ via a bridge (see \cref{fig!pmtsehard}).
  We ask for constructing roughly $n$ paths connecting the terminals, where $n$ is the number of vertices in the input graph of PCHC.
  All but one of these paths will use the bridge to~$t$ in the constructed graph for~$n$ time steps in total, each in a different time step. 
  Thus, this bridge is occupied for roughly $n$ time steps, and the final path is forced to stay in the input graph of PCHC for $n$ time steps.
  For a path, this is only possible by visiting each of the $n$ vertices in the graph exactly once, and hence it corresponds to a \hcycle{}.

The reduction to prove \cref{thm!pmtsehard} uses the following \cref{constr:pmtsehard}.

\begin{constr}
 \label{constr:pmtsehard}
 Let $G=(V,E)$ be an undirected, planar, cubic graph with $n=|V|$.
 Construct in time polynomial in the size of~$G$ an undirected planar graph~$G'$ as follows (refer to \cref{fig!pmtsehard}\iflong{} for an illustration of the constructed graph\fi{}).
 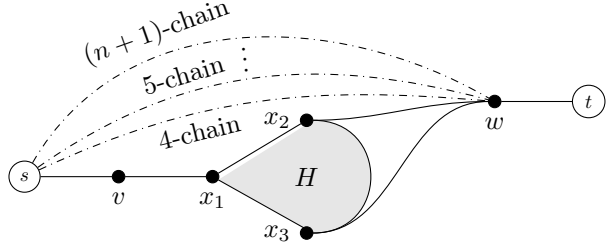
\begin{figure}[t!]
  \centering
 \begin{tikzpicture}

  \tikzstyle{xnode}=[circle, fill, scale=1/2, draw];
  \def\xsc{1.25}
  \node (s) at (0*\xsc,-1)[circle, scale=3/4, draw]{$s$};
  \node (v) at (1*\xsc,-1)[xnode,label=-90:{$v$}]{};

  \node (x1) at (2*\xsc,-1)[xnode,label=-90:{$x_1$}]{};
  \node (x2) at (3*\xsc,-0.25)[xnode,label=180:{$x_2$}]{};
  \node (x3) at (3*\xsc,-1.75)[xnode,label=180:{$x_3$}]{};

  \node (w) at (5*\xsc,0)[xnode,label=-90:{$w$}]{};
  \node (t) at (6*\xsc,0)[circle, scale=3/4, draw]{$t$};

  \draw[dashdotted] (s) to [out=25, in=175] node[pos=.35,below,sloped] {$4$-chain}(w);
  \draw[dashdotted] (s) to [out=35, in=165] node[pos=.35,above,sloped] {$5$-chain}(w);
  \path (s) to [out=50, in=155]  node[sloped] {$\vdots$} (w);
  \draw[dashdotted] (s) to [out=60, in=150]  node[pos=.35,above,sloped] {$(n+1)$-chain} (w);

  \begin{pgfonlayer}{background}
    \draw[fill=gray!20, draw] (x1) to (x2) to [out=0,in=0,looseness=1.75](x3) to (x1);
  \end{pgfonlayer}
  \node (Gh) at (3*\xsc,-1)[]{$H$};

  \draw (x2) to [out=0,in=180](w);
  \draw (x3) to [out=0,in=180](w);

  \draw (w) to (t);
  \draw (s) to (v);
  \draw (v) to (x1);

  \end{tikzpicture}
   \caption{Graph~$G'$ obtained in~\cref{constr:pmtsehard}. The gray part represents the graph~$H$. Dashed lines represent chains.}\label{fig!pmtsehard}
 \end{figure}
 Let initially~$G'$ be the empty graph.
 Add a copy of~$G$ to~$G'$.
 Denote the copy of~$G$ in~$G'$ by~$H$.
 Next, add the new vertices~$s,t,v,w$ to~$G$.
 Connect~$s$ with~$v$, and~$w$ with~$t$ by an edge.
 For each~$m\in\{4,5,\ldots,n+1\}$, add an $m$-chain connecting~$s$ with~$w$.
 Next, consider a fixed plane embedding~$\phi(G)$ of~$G$.
 Let~$x_1$ denote a vertex incident to the outer face in~$\phi(G)$.
 Then, there are two neighbors~$x_2$ and~$x_3$ of~$x_1$ also incident to the outer face in~$\phi(G)$.
 Add the edges~$\{v,x_1\}$,~$\{x_2,w\}$ and~$\{x_3,w\}$ to~$G'$ completing the construction of~$G'$.
 We remark that~$G'$ is planar as it allows a plane embedding (see \cref{fig!pmtsehard}) using $\phi$ as an embedding of~$H$.
 \qed
\end{constr}

\begin{lemma}
 \label{lem:pmtsehard}
 Let $G$ and~$G'$ be as in~\cref{constr:pmtsehard}.
 Then~$G$ admits a \hcycle{} if and only if~$G'$ allows for at least~$n-1$ $s$-$t$~paths with no \shared edge.
\end{lemma}

\begin{proof}
 \LD{} %
 Let $\calP$ denote a set of $n-1$ $s$-$t$~paths in $G'$ with no \shared edge.
 Note that the degree of $s$ is equal to $n-1$.
 As no two paths in~$\calP$ \share any edge in~$G'$, each path in~$\calP$ uses a different edge incident to~$s$.
 This implies that $n-2$ paths in $\calP$ uniquely contain each of the chains connecting $s$ with $w$, and one path~$P \in \calP$ contains the edge $\{s,v\}$.
 Note that each of the $n-2$ paths contain the vertex~$w$ at most once, and since they contain the chains connecting~$s$ with~$w$, the edge~$\{w,t\}$ appears at the time~steps~$\{5,6,\ldots,n+2\}$ in these $n - 2$ paths~$\calP$. 
 Hence, the path~$P$ has to contain the edge~$\{w,t\}$ at a time~step smaller than five or larger than~$n+2$. 
 Observe that, by construction, the shortest path between $s$ and $w$ is of length 4 and, thus, $P$ cannot contain the edge~$\{w,t\}$ on any time~step smaller than five.
 Hence, $P$ has to contain the edge at time~step at least $n+3$. 
 Since the distance between $s$ and $x_1$ is two, and the distance from $x_2$, $x_3$ to $w$ is one, $P$ has to visit each vertex in~$H$ exactly once, starting at $x_1$, and ending at one of the two neighbors $x_2$ or $x_3$ of $x_1$.
 Hence, $P$ restricted to $H$ describes a Hamiltonian path in $H$, which can be extended to an Hamiltonian cycle by adding the edge $\{x_1,x_2\}$ in the first or $\{x_1,x_3\}$ in the second case.
 
 \RD{} 
 \looseness=-1 Let $G$ admit a Hamiltonian cycle~$C$.
 Since $C$ contains every vertex in~$G$ exactly once, it contains $x_1$ and its neighbors $x_2$ and $x_3$. 
 Since $C$ forms a cycle in~$G$ and~$G$ is cubic, at least one of the edges $\{x_1,x_2\}$ or $\{x_1,x_3\}$ appears in~$C$.
 Let $C'$ denote an ordering of the vertices in $C$ such that $x_1$ appears first and the neighbor $x\in \{x_2,x_3\}$ of $x_1$ with $\{x_1,x\}$ contained in $C$ appears last.
 We construct $n-1$ $s$-$t$ paths without \sharing an edge. 
 First, we construct $n-2$ $s$-$t$ paths, each containing a different chain connecting $s$ with $w$ and the edge $\{w,t\}$.
 Observe that since the lengths of each chain is unique, no edge (in particular, not $\{w,t\}$) is \shared.
 Finally, we construct the one remaining $s$-$t$ path~$P$ as follows.
 We lead~$P$ from~$s$ to~$x_1$ via~$v$, then following~$C'$ in~$H$ to~$x$, and then from~$x$ to~$t$ via~$w$.
 Observe that~$P$ has length $n+3$ and contains the edge $\{w,t\}$ at time~step~$n+3$.
 Hence, no edge is \shared as the path containing the $(n+1)$-chain contains the edge~$\{w,t\}$ at time~step~$(n+2)$.
 We constructed $n-1$ $s$-$t$~paths in~$G'$ with no \shared edge.
 \ifshort{}\qed\fi{}
\end{proof}

\toappendix
{
  Note that the maximum degree of the graph obtained in the~\cref{constr:pmtsehard} depends on the number of vertices in the input graph.
  In what follows, we give a second construction where the obtained graph has constant maximum degree~$\Delta=4$.

  \begin{constr}
  \label{constr:pmtsehardconstdeg}
  Let $G=(V,E)$ be an undirected, planar, cubic graph with $n=|V|$.
  Construct in time polynomial in the size of~$G$ an undirected planar graph~$G'$ as follows (refer to \cref{fig!pmtsehardconstdeg} for an illustration of the constructed graph).
  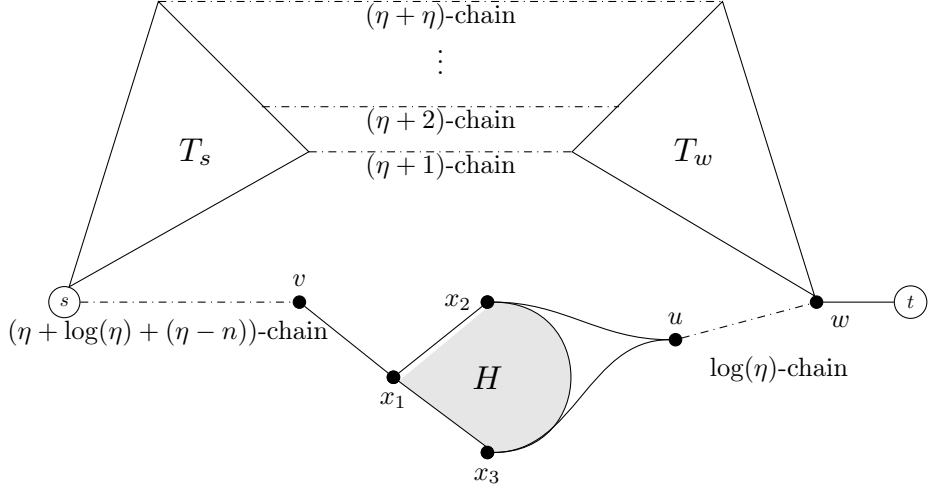
\begin{figure}[t!]
  \centering
    \begin{tikzpicture}

    \tikzstyle{xnode}=[circle, fill, scale=1/2, draw];
    \def\xsc{1.25}
    \node (s) at (-1.5*\xsc,0)[circle, scale=3/4, draw]{$s$};
    \node (v) at (1*\xsc,0)[xnode,label=90:{$v$}]{};

    \node (x1) at (2*\xsc,-1)[xnode,label=-90:{$x_1$}]{};
    \node (x2) at (3*\xsc,0)[xnode,label=180:{$x_2$}]{};
    \node (x3) at (3*\xsc,-2)[xnode,label=-90:{$x_3$}]{};

  \node (wp) at (5*\xsc,0-0.5)[xnode,label=90:{$u$}]{};
    \node (w) at (6.5*\xsc,0)[xnode,label=-45:{$w$}]{};
    \node (t) at (7.5*\xsc,0)[circle, scale=3/4, draw]{$t$};

    \draw[fill=gray!20, draw] (x1) to (x2) to [out=0,in=0,looseness=1.75](x3) to (x1);
    \node (x1) at (2*\xsc,-1)[xnode]{};
    \node (x2) at (3*\xsc,0)[xnode]{};
    \node (x3) at (3*\xsc,-2)[xnode]{};
    \node (Gh) at (3*\xsc,-1)[scale=1.25]{$H$};

    \draw (x2) to [out=0,in=180](wp);
    \draw (x3) to [out=0,in=180](wp);

  \draw[dashdotted] (wp) to (w);
	  \node (txt) at (6.1*\xsc,0-0.9)[]{$\log(\varn)$-chain};
    \draw (w) to (t);
    \draw[dashdotted] (s) to (v);
	  \node (txt) at (-0.4*\xsc,0-0.4)[]{$(\varn+\log(\varn)+(\varn-n))$-chain};
    \draw (v) to (x1);

  \draw (s) -- (-0.5*\xsc,4) -- (1.1*\xsc,2.0) -- cycle;
  \draw (w) -- (5.5*\xsc,4) -- (3.9*\xsc,2.0) -- cycle;
    \node (txt) at (-0.1*\xsc,2)[scale=1.25]{$T_s$};
    \node (txt) at (5.2*\xsc,2)[scale=1.25]{$T_w$};
  \draw[dashdotted]  (1.1*\xsc,2.0)  to (3.9*\xsc,2.0);
  \draw[dashdotted]  (0.6*\xsc,2.6)  to (4.4*\xsc,2.6);
  \draw[dashdotted] (-0.5*\xsc,4) to (5.5*\xsc,4);
    \node (txt) at (2.5*\xsc,1.8)[]{$(\varn+1)$-chain};
    \node (txt) at (2.5*\xsc,2.4)[]{$(\varn+2)$-chain};
    \node (txt) at (2.5*\xsc,3.3)[]{$\vdots$};
    \node (txt) at (2.5*\xsc,3.8)[]{$(\varn+\varn)$-chain};

    \end{tikzpicture}
  \caption{$T_s$ and $T_w$ refer to complete binary trees with $\varn$ leaves.}
  \label{fig!pmtsehardconstdeg}
  \end{figure}
  Let initially~$G'$ be the graph obtained from~\cref{constr:pmtsehard}.
  Remove $s$, $w$, and the chains connecting~$s$ with~$w$ from~$G'$.
  Add a vertex~$u$ to~$G'$ and add the edges~$\{x_2,u\}$ and~$\{x_3,u\}$ to~$G'$.
  Let $\varn$ be the smallest power of two larger than~$n$ (note that $n\leq \varn\leq 2n-2$).
  Add a complete binary tree~$T_s$ with $\varn$~leaves to~$G'$, and denote its root by~$s$.
  Denote the leaves by $a_1,\ldots,a_\varn$, ordered by a post-order traversal on~$T_s$.
  Next add a copy of~$T_s$ to~$G'$, and denote the copy by~$T_w$ and its root by~$w$. 
  If~$a_i$ is a leaf of~$T_s$, denote by $a_i'$ its copy in~$T_w$.
  Next, for each $i\in[\varn]$, connect~$a_i$ and~$a_i'$ via a~$(\varn+i)$-chain.
  Finally, connect~$s$ with~$v$ via an $(\varn+\log(\varn)+(\varn-n))$-chain, connect~$u$ with~$w$ via a $\log(\varn)$-chain, and add the edge~$\{w,t\}$ to~$G'$, which completes the construction of~$G'$.
  Note that~$T_s$ and~$T_w$ allow plane drawings, and the chains connecting the leaves can be aligned as 
  illustrated in~\cref{fig!pmtsehardconstdeg}. 
  It follows that~$G'$ allows for a plane embedding.
  \qed
  \end{constr}
    
  \begin{lemma}%
  \label{lem:pmtsehardconstdeg}
  Let $G$ and~$G'$ be as in~\cref{constr:pmtsehardconstdeg}.
  Then~$G$ admits a \hcycle{} if and only if~$G'$ allows for at least~$\varn+1$ $s$-$t$~paths with at most~$\varn-2$ \shared edges.
  \end{lemma}

  {
  \begin{proof}
  \RD{}
  Let~$G$ admit a \hcycle{}~$C$.
  As discussed in the proof of~\cref{lem:pmtsehard}, there is an ordering~$C'$ of the vertices in~$C$ such that~$x_1$ appears first and $x\in\{x_2,x_3\}$ appears last in~$C'$.
  We construct~$\varn$ $s$-$t$~paths as follows.
  We route them from $s$ in $T_s$ to the leaves of $T_s$ in such a way that each path contains a different leaf of~$T_s$.
  Herein, $\varn-2$ edges are \shared.
  Next, route each of them via the chain connecting the leaf to the corresponding leaf in~$T_w$, then to~$w$, and finally to~$t$.
  In this part, no edge is \shared, as the lengths of the chains connecting the leaves of~$T_s$ and~$T_w$ are pairwise different.
  Hence, the $\varn$ $s$-$t$~paths contain the edge~$\{w,t\}$ at the time steps~$2\log(\varn)+\varn+i+1$ for each~$i\in[\varn]$.
  We construct the one remaining $s$\nobreakdash-$t$~path~$P$ as follows.
  The path~$P$ contains the chain connecting $s$ with $v$, the edge~$\{v,x_1\}$.
  Then~$P$ follows~$C'$ in~$H$ to $x\in\{x_2,x_3\}$, via the edge~$\{x,u\}$ to~$u$, via the chain connecting~$u$ with~$w$ to~$w$, and finally to~$t$ via the edge~$\{w,t\}$.  
  Observe that $P$ contains the edge~$\{w,t\}$ at the time step~$2\varn+2\log(\varn)+2$, and hence no \sharing any further edge in~$G'$.
  
  \LD{}
  Let~$\calP$ be a set of $\varn+1$ $s$-$t$~paths in~$G'$ \sharing at most~$k:=\varn-2$ edges.
  First observe that no two paths contain the chain connecting~$s$ with~$v$, as otherwise more that $k$ edges are \shared.
  Hence, at most one $s$-$t$~path leaves~$s$ via the chain to~$v$.
  It follows that at least~$\varn$ paths leave~$s$ via the edges in $T_s$. 
  By the definition of paths, observe that each of $s$-$t$~paths arrive at a leaf of~$T_s$ at same time step.
  Suppose at least two $s$-$t$ paths contain the same leaf of~$T_s$.
  As each leaf is of degree two, the $s$-$t$ paths follow the chain towards a leaf of~$T_w$ simultaneously.
  This introduces at least $\varn+1>k$ \shared edges, 
  contradicting the choice of~$\calP$.
  It follows that exactly~$\varn$ $s$-$t$~paths leave~$s$ via~$T_s$ (denote the set by~$\calP'$), and they arrive each at a different leaf of~$T_s$ at time~step~$\log(\varn)$.
  Moreover, by construction, each path in~$\calP'$ arrives at a different leaf of $T_w$ at the time~steps~$\log(\varn)+\varn+i+1$ for every~$i\in[\varn]$.
  
  We next discuss why no path in~$\calP'$ contains more than one chain connecting a pair~$(a,a')$ of leaves, where~$a$ and~$a'$ are leaves of~$T_s$ and~$T_w$, respectively.
  Assume that there is a path~$P'\in\calP'$ containing at least two chains connecting the pairs~$(a,a')$ and~$(b,b')$ of leaves, where~$a,b$ and~$a',b'$ are leaves of~$T_s$ and~$T_w$, respectively, and vertex~$a$ appears at smallest time~step over all such leaves of~$T_s$ in~$P'$ and $b'$ appears at smallest time~step over all such leaves of~$T_w$ in~$P'$ (recall that~$P'$ must contain a leaf of~$T_s$ at smaller time~step than every leaf in~$T_w$).
  By construction, $a'$ and $b'$ are the copies of~$a$ and~$b$ in~$T_w$.
  Let~$r$ denote the vertex in~$T_s$ such that $r$ is the root of the subtree of minimum height in~$T_s$ containing~$a$ and~$b$ as leaves. 
  Let $r'$ denote its copy in~$T_w$.
  Observe that by construction, $r'$~is the root of the subtree of minimum height in~$T_w$ containing~$a'$ and~$b'$ as leaves.
  As $P'$ starts at vertex~$s$, and the path from $s$ to $a$ in~$T_s$ is unique, $P'$ contains the vertex~$r$ at smallest time step among the vertices in~$X:=\{r,a,a',r',b',b\}$.
  As vertex~$a$ appears at smallest time step over all leaves of~$T_s$ and $a$ is of degree two, $P'$ contains the vertices~$a$ and $a'$ at second and third smallest time step, respectively, among the vertices in~$X$.
  As in any tree, the unique path between every two leaves contains the root of the subtree of minimum height containing the leaves, $P'$ contains the vertex~$r'$ at fourth smallest time step among the vertices in~$X$. 
  Finally, as vertex~$b'$ appears at smallest time step over all leaves of~$T_w$ and $b'$ is of degree two, $P'$ contains the vertices~$b'$ and $b$ at fifth and sixth smallest time steps, respectively, among the vertices in~$X$.
  In summary, the vertices in~$X$ appear in~$P'$ in the order~$(r,a,a',r',b',b)$.
  Now, observe that~$\{r,r'\}$ forms a $b$-$t$ separator in~$G'$, that is, there is no $b$-$t$~path in~$G'-\{r,r'\}$.
  As~$P'$ contains~$r$ and~$r'$ at smaller time steps than~$b$, $P'$ contains a vertex different to~$t$ at last time~step.
  This contradict the fact that $P'$ is an $s$-$t$~path in $G'$.
  It follows that no path in~$\calP'$ contains more than one chain connecting a pair consisting of leaf of~$T_s$ and a leaf of~$T_w$.

  It follows that the paths in~$\calP'$ contain the edge~$\{w,t\}$ at the time steps~$2\log(\varn)+\varn+i+1$ for each~$i\in[\varn]$.
  Hence, the remaining $s$-$t$~path containing the chain connecting~$s$ with~$v$, denoted by~$P$, has to contain the edge~$\{w,t\}$ at time~step at most~$2\log(\varn)+\varn+1$ or at 
  least~$2\log(\varn)+2\varn+2$.
  At the earliest $P$ can contain the edge~$\{w,t\}$ at time~step~$2\log(\varn)+\varn+(\varn-n)+3> 2\log(\varn)+\varn+1$, and thus, path~$P$ has to contain edge~$\{w,t\}$ at 
  time~step~$2\log(\varn)+2\varn+2$.
  This is only possible if~$P$ forms a path in~$H$ that visits each vertex in~$H$, starting at~$x_1$ and ending at vertex~$x\in\{x_2,x_3\}$. 
  As the edge~$\{x_1,x\}$ is contained in~$H$, it follows that~$P$ restricted to~$H$ forms a \hcycle{} in~$G$.
  \ifshort{}\qed\fi{}
  \end{proof}
  }%
}

\ifshort{}
\noindent The remaining proof of \cref{thm!pmtsehard} is deferred to \cref{appsec:sec:pmtse}. We remark that the statement in \cref{thm!pmtsehard} for constant~$k$ follows from~\cref{lem:pmtsehard}.

\fi{}
\appendixproof{thm!pmtsehard}
{
  \begin{proof}[\iflong{}Proof of~\cref{thm!pmtsehard}\fi{}]
    We provide a many-one reduction from \textsc{PCHC} to \pmtseAcr on undirected graph via \cref{constr:pmtsehard} (for constant number~$k$ of shared edges) on the one hand, and~\cref{constr:pmtsehardconstdeg} (for constant maximum degree~$\Delta$) on the other.
  Let $(G)$ be an instance of~\textsc{PCHC} with $n=|V(G)|$.
  
  \smallskip\noindent\emph{Via \cref{constr:pmtsehard}.}
  Let $(G',s,t,p,0)$ be an instance of~\pmtseAcr where $G'$ is obtained from~$G$ by applying~\cref{constr:pmtsehard} and $p=n-1$.
  Note that~$(G',s,t,p,0)$ is constructed in polynomial time and, by~\cref{lem:pmtsehard}, $G$ is a yes-instance of~\textsc{PCHC} if and only if~$(G',s,t,p,0)$ is a yes-instance of~\pmtseAcr.
  
  \emph{The case of constant~$k>0$.}
  Reduce $(G',s,t,p,0)$ to an equivalent instance $(G_k',s',t,p,k)$ of~\pmtseAcr with $k>0$ as follows.
  Let~$G_k'$ denote the graph obtained from~$G'$ by the following modification: Add a chain of length~$k$ to~$G'$, and identify one endpoint with~$s$ and denote by~$s'$ the other endpoint.
  Set~$s'$ as the new source.
  Observe that any $s'$-$t$ path in~$G_k'$ contains the $k$-chain appended on~$s$, and hence, any solution introduces exactly~$k$ \shared edges.
  
  \emph{The directed case.}
  Direct the edges in $G'$ as follows.
  Direct each chain connecting~$s$ with~$w$ from~$s$ towards~$t$. (In the case of $k>0$, also direct the chain from~$s'$ towards~$s$.)
  Direct the edges~$\{s,v\}$, $\{v,x_1\}$, $\{x_2,w\}$, $\{x_3,w\}$, and $\{w,t\}$ as $(s,v)$, $(v,x_1)$, $(x_2,w)$, $(x_3,w)$, and $(w,t)$.
  Finally, replace each edge~$\{a,b\}$ in~$H$ by two (anti-parallel) arcs $(a,b),(b,a)$ to obtain the directed variant of~$H$.
  The correctness follows from the fact that we consider paths that are not allowed to contain vertices more than once.
  Note that the planarity is not destroyed.

  {
      \smallskip\noindent\emph{Via \cref{constr:pmtsehardconstdeg}.}
      Let $(G',s,t,p,k)$ be an instance of~\pmtseAcr where $G'$ is obtained from~$G$ by applying~\cref{constr:pmtsehardconstdeg}, $p=\varn+1$, and $k=\varn-2$.
      Note that~$(G',s,t,p,k)$ is constructed in polynomial time and, by~\cref{lem:pmtsehardconstdeg}, $G$ is a yes-instance of~\textsc{PCHC} if and only if~$(G',s,t,p,k)$ is a yes-instance of~\pmtseAcr.
      
      \emph{The directed case.}
      Direct the edges in $G'$ as follows.
      Direct the edges in $T_s$ from $s$ towards the leaves, and the edges in $T_w$ from the leaves towards~$w$.
      Direct each chain connecting~$T_s$ with~$T_w$ from~$T_s$ towards~$T_w$.
      Direct the edges~$\{v,x_1\}$, $\{x_2,u\}$, $\{x_3,u\}$, and $\{w,t\}$ as~$(v,x_1)$, $(x_2,u)$, $(x_3,u)$, and $(w,t)$.
      Direct the chain connecting~$s$ with~$v$ from~$s$ towards~$v$, and the chain connecting~$u$ with~$w$ from~$u$ towards~$w$.
      Finally, replace each edge~$\{a,b\}$ in~$H$ by two (anti-parallel) arcs $(a,b),(b,a)$ to obtain the directed variant of~$H$.
      The correctness follows from the fact that we consider paths that are not allowed to contain vertices more than once.
      Note that the planarity is not destroyed.
    }
  \ifshort{}\qed\fi{}%
  \end{proof}
}
As the length of every $s$-$t$ path is upper bounded by the number of vertices in the graph, we immediately obtain the following.

\begin{corollary}
 \label{cor:psmtsehard}
 \psmtseAcr{} both on undirected planar and directed planar graphs is \NP-complete, even if~$k\geq0$ is constant or $\Delta\geq4$~is constant.
\end{corollary}

\section{Trail-\mtseAcr}
\label{sec:tmtse}
\appendixsection{sec:tmtse}

\looseness=-1 We now show that \tmtseAcr{} has the same computational complexity fingerprint as \pmtseAcr{}.
That is, \tmtseAcr{} (\tsmtseAcr{}) is \NP-complete on undirected and directed planar graphs, even if the number~$k\geq0$ of shared edges (arcs) or the maximum degree~$\Delta\geq 5$ ($\Delta_{\rm i/o}\geq 3$) is constant. The reductions are slightly more involved, because it is harder to force trails to take a certain way.  

\subsection{On Undirected Graphs}
\label{ssec:tmtse:undir}

\iflong{}In this section, we prove the following.\fi{}

\begin{theorem}%
  \label{thm!tmtseundirhard}
 \tmtseAcr on undirected planar graphs is \NP-complete, even if~$k\geq 0$ is constant or~$\Delta\geq 5$ is constant.
\end{theorem}
We provide two constructions supporting the two subresults for constants~$k, \Delta$. The reductions are again from \textsc{Planar Cubic Hamiltonian Cycle} (PCHC).

{
  \begin{constr}
  \label{constr:tmtseundirhard}
  Let $G=(V,E)$ be an undirected planar cubic graph with $n=|V|$.
  Construct an undirected planar graph~$G'$ as follows (refer to \cref{fig:undirtrailmtse}\iflong{} for an illustration of the constructed graph\fi{}).
  \begin{figure}[t]\centering
  \begin{tikzpicture}

	\tikzstyle{xnode}=[circle, fill, scale=1/2, draw];
	\def\xsc{0.95}
	\node (s) at (0*\xsc,0)[circle, scale=4/5, draw]{$s$};
	\node (v) at (2*\xsc,-.5)[xnode,label=-20:{$v$}]{};
	\node (x1) at (2*\xsc,0)[xnode,label=-90:{$x$}]{};
	\node (x2) at (2*\xsc,1.5)[scale=0.1]{};
	\node (w) at (4*\xsc,0)[xnode,label=-90:{$w$}]{};
	\node (t) at (6*\xsc,0)[circle, scale=4/5, draw]{$t$};

	\draw[fill=gray!20, draw] (x1) to [out=0,in=0,looseness=1.75](x2) to [out=0+180,in=+180,looseness=1.75](x1);
	\node (x1) at (2*\xsc,0)[xnode]{};
	\node (x2) at (2*\xsc,1.5)[scale=0.1]{};
	\node (Gh) at (2*\xsc,0.75)[]{$H'$};

	\draw[-]  (x1) to (w);
	\draw[-] (w) to (t);
	\draw[-] (s) to (v);
	\draw[-] (s) to (x1);
	\draw[-] (v) to (w);

      \node (l1) at (-0.75*\xsc,0)[xnode]{};
      \draw (s) to[out=-165,in=-35](l1);
      \draw (s) to[out=165,in=35](l1);
      \node (l1) at (-1.5*\xsc,0)[xnode]{};
      \draw (s) to[out=-155,in=-35](l1);
      \draw (s) to[out=155,in=35](l1);
      \node (ld) at (-2.25*\xsc,0)[]{$\ldots$};
      \node (l1) at (-3*\xsc,0)[xnode]{};
      \draw (s) to[out=-145,in=-30](l1);
      \draw (s) to[out=145,in=30](l1);

      \draw[decorate,decoration={brace,amplitude=10pt}] (-3*\xsc,1.4-0.75) -- (-0.5*\xsc,1.4-0.75);
      \node (txt) at (-1.75*\xsc,2.05-0.75)[]{$n-1$};

    \end{tikzpicture}
  \caption{Graph~$G'$ obtained in~\cref{constr:tmtseundirhard}. The gray part represents the graph~$H'$.}
  \label{fig:undirtrailmtse} 
  \end{figure}
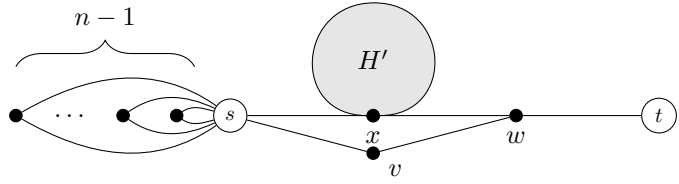
  Initially, let $G'$ be the empty graph.
  Add a copy of~$G$ to~$G'$ and denote the copy by~$H$.
  Subdivide each edge in~$H$ and denote the resulting graph~$H'$.
  Note that $H'$ is still planar. 
  Consider a plane embedding~$\phi(H')$ of~$H'$ and and let~$x\in V(H')$ be a vertex incident to the outer face in the embedding.
  Next, add the vertex set~$\{s,v,w,t\}$ to~$G$.
  Add the edges~$\{s,x\}$, $\{s,v\}$, $\{v,w\}$, and~$\{w,t\}$ to~$G$.
  Finally, add $n-1$ vertices~$B=\{b_1,\ldots,b_{n-1}\}$ to~$G$ and connect each of them with $s$ by two edges (in the following, we distinguish these edges as~$\{s,b_i\}_1$ and~$\{s,b_i\}_2$, for each~$i\in[n-1]$).
  Note that the graph is planar (see \cref{fig:undirtrailmtse} for an embedding, where $H'$ is embedded as~$\phi(H')$) but not simple.
  \qed
  \end{constr}

  \begin{lemma}%
  \label{lem:tmtseundirhard}
  Let $G$ and $G'$ as in \cref{constr:tmtseundirhard}.
  Then $G$ admits a \hcycle{} if and only if~$G'$ admits~$2n$ $s$-$t$~trails with no \shared edge.
  \end{lemma}

   \appendixproof{lem:tmtseundirhard}
  {
    \begin{proof}
    \RD{}
    Let~$G$ admit a \hcycle{}~$C$.
    Observe that~$H'$ allows for a cycle~$C'$ in~$H'$ that contains each vertex corresponding to a vertex in~$H$ exactly once.
    We construct~$2n$ trails in~$G'$ as follows.
    
    We group the trails in two groups.
    The first group of trails first visits some of the vertices~$b_1,\ldots,b_{n-1}$ by each time first using the edge~$\{s,b_i\}_1$, $i\in[n-1]$, and then proceeding to~$t$ via~$v$.
    The second group of trails first visits some of the vertices~$b_1,\ldots,b_{n-1}$ by each time first using the edge~$\{s,b_i\}_2$, $i\in[n-1]$, and then proceeding via~$x$, then following the cycle~$C'$, and finally again via~$x$ towards~$t$.
    Let $T_1^i,\ldots,T_n^i$ denote the trails of group~$i\in\{1,2\}$.
    For each $j<n$, the trail $T^i_j$ first visits the vertices $b_j,\ldots,b_{n-1}$ in that order before proceeding as described above.
    The trails $T_n^i$, $i\in\{1,2\}$, do not contain any of the vertices~$b_1,\ldots, b_{n-1}$, and directly approach~$t$ as described above.
    
    Observe that, within each of the two groups, no two trails \share an edge.
    Between trails of different groups, only edge~$\{w, t\}$ can possibly be \shared.
    Note that any cycle in~$H'$ is of even length. 
    Hence, the trails of group~1 contain the edge~$\{w,t\}$ at each of the time steps~$2j+1$ for every~$j\in[n]$.
    The trails of group~2 contain the edge~$\{w,t\}$ at each of the time steps~$2n+2j+1$ for every~$j\in[n]$.
    Hence, no two trails \share an edge.
    
    \LD{}
    Let $G'$ admit $2n$ $s$-$t$~trails with no \shared edge.
    First note that $s$ has exactly $2n$ incident edges.
    Observe that, for each $\beta = 0, \ldots, |B|$, no more than two trails contain $\beta$ vertices of~$B$, as otherwise any of the edges~$\{s,v\}$ or~$\{s,x\}$ would be \shared.
    By the pigeon hole principle it follows that, for each $\beta = 0, \ldots, |B|$, there are exactly two trails that contain $\beta$ vertices of~$B$.
    Hence, for each even time step, there are two trails leaving $s$ via the edges $\{s,v\}$ and $\{s,x\}$, respectively. 
    Observe that those trails that proceed towards $t$ via~$v$ use the edges~$\{w,t\}$ exactly at the time steps $2j+1$ for every~$j\in[n]$.
    Because each trial in $H'$ that starts and ends at the same vertex has even length, those trials that proceed towards $t$ via~$x$ can use $\{w, t\}$ only at odd time steps.
    Hence, since the edge $\{w,t\}$ is not \shared, the trails proceeding towards~$t$ via~$x$ need to stay in~$H'$ for~$2n$ time steps.
    As $H'-x$ has maximum degree three, no vertex in~$H'$ beside~$x$ is contained more than once in all of these trails.
    As the length between every two vertices in~$H'$ corresponding to vertices in~$H$, it follows that every of these trails visits the vertices in~$H'$ corresponding to the vertices in~$H$.
    It follows that each of these trails forms a Hamiltonian cycle~$C'$ in~$H'$.
    As $C'$ can easily turned into a Hamiltonian cycle~$C$ in~$G$ (consider the sequence when deleting all vertices that do not correspond to a vertex in~$H$), the statement follows.
      \ifshort{}\qed\fi{}
    \end{proof}
  }%
  To deal with the parallel edges in graph~$G'$ in \cref{constr:tmtseundirhard}, we now subdivide edges, maintaining an equivalent statement as in~\cref{lem:tmtseundirhard}.
  \begin{lemma}%
  \label{lem:doublesubdiv}
  \looseness=-1 Let $G$ be an undirected graph (not necessarily simple) \iflong{}with two distinct vertices $s$ and $t$\fi{}.
  Obtain graph~$G'$ from~$G$ by replacing each edge~$\{u,v\}\in E$ in~$G$ by a path of length three, identifying its endpoints with~$u$ and~$v$.
  Then $G$ admits $p\in \N$ $s$-$t$ trails with no \shared edge if and only if $G'$ admits $p$ $s$-$t$ trails with no \shared edge.
  \end{lemma}
  \appendixproof{lem:doublesubdiv}
  {
    \begin{proof}
      For each edge~$e\in E(G)$ in $G$ denote by~$P(e)$ the corresponding path of length three in~$G'$ .
      By definition, for all $e,f\in E(G')$ it holds that $e\neq f$ if and only if $P(e)\neq P(f)$.
      
    \RD{}
	Let $\calP=\{P_i\mid i\in [p]\}$ be a set of $p$ $s$-$t$ trails in~$G$ with no \shared edge.
	Each $P_i$ as an edge sequence representation~$P_i=(e_1^i,\ldots,e_{\ell_i}^i)$, where $\ell_i$ is the number of edges in~$P_i$.
	For each~$P_i$, consider the corresponding trail~$P_i'=(P(e_1^i),\ldots,P(e_{\ell_i}^i))$ in~$G'$, and the set~$\calP'=\{P_i'\mid i\in[p]\}$.
	Suppose that two trails~$P_i'$ and~$P_j'$ \share an edge.
	Then the \shared edge is contained in subpaths~$P(e_x^i)$ and $P(e_x^j)$.
	As $P(e_x^i)$ and $P(e_x^j)$ are not edge-disjoint (as they \share an edge), it follows that $e_x^i=e_x^j$, and hence $P_i$ and $P_j$ \share the edge~$e_x^i$ in $G$.
	This contradicts the fact that $\calP=\{P_i\mid i\in [p]\}$ is a set of $p$ $s$-$t$ trails in~$G$ with no \shared edge.
	It follows that~$\calP'$ is a set of $p$ $s$-$t$ trails in~$G'$ with no \shared edge.

    \LD{}
      Let $\calP'=\{P_i'\mid i\in [p]\}$ be a set of $p$ $s$-$t$ trails in~$G'$ with no \shared edge.
      Observe that, by the construction of~$G'$, each $s$-$t$ trail in~$G'$ is composed of paths of length three with endpoints corresponding to vertices in~$G$.
      Hence, for each $i\in[p]$, let $P_i'$ be represented as~$P_i'=(P(e_1^i),\ldots,P(e_{\ell_i}^i))$, where~$3\cdot \ell_i$ is the number of edges in~$P_i'$.
      For each~$P_i'$, consider the corresponding trail~$P_i=(e_1^i,\ldots,e_{\ell_i}^i)$ in~$G$, and the set~$\calP=\{P_i\mid i\in[p]\}$.
      Suppose that two trails~$P_i$ and $P_j$ \share an edge, that is, there is an index~$x$ such that~$e_x^i=e_x^j$.
      It follows that $P(e_x^i)=P(e_x^j)$.
      Let~$e_x^i=\{v,w\}=:e$.
      If both trails~$P_i'$ and~$P_j'$ traverse~$P(e)$ in the same ``direction'', i.e.~either from~$v$ to~$w$ or from~$w$ to~$v$, then $P_i'$ and $P_j'$ \share at least three edges (all edges in~$P(e)$).
      This contradicts the definition of~$\calP'$.
      Consider the case that the trails~$P_i'$ and~$P_j'$ traverse~$P(e)$ in opposite ``directions'', i.e.~one from~$v$ to~$w$ and the other from~$w$ to~$v$.
      As $P(e)$ is of length three, the edge in~$P(e)$ with no endpoint in~$\{v,w\}$ is then used by~$P_i'$ and~$P_j'$ at the same time step, yielding that the edge is \shared.
      This contradicts the definition of~$\calP'$.
      It follows that $\calP=\{P_i\mid i\in[p]\}$ is a set of $p$ $s$-$t$ trails in~$G$ with no \shared edge.
    \ifshort{}\qed\fi{}
    \end{proof}
  }%
}
We now show how to modify \cref{constr:tmtseundirhard} for maximum degree five, giving up, however, a constant upper bound on the number of shared edges.

\begin{constr}
 \label{constr:tmtseundirhardconstdeg}
\looseness=-1 Let $G=(V,E)$ be an undirected planar cubic graph with $n=|V|$.
 Construct an undirected planar graph~$G'$ as follows (see \cref{fig:undirtrailmtseconstdeg}\iflong{} for an illustration of the constructed graph\fi{}).
 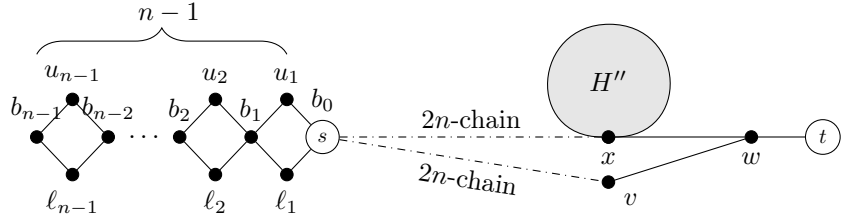
\begin{figure}[t!]
  \centering
  \begin{tikzpicture}

    \usetikzlibrary{arrows}
    \usetikzlibrary{decorations.pathreplacing}

      \tikzstyle{xnode}=[circle, fill, scale=1/2, draw];
      \def\xsc{0.95}
      \node (s) at (0*\xsc,0)[circle, scale=4/5, draw,,label=90:{$b_0$}]{$s$};
      \node (v) at (4*\xsc,-.60)[xnode,label=-20:{$v$}]{};
      \node (x1) at (4*\xsc,0)[xnode,label=-90:{$x$}]{};
      \node (x2) at (4*\xsc,1.5)[scale=0.1]{};
      \node (w) at (6*\xsc,0)[xnode,label=-90:{$w$}]{};
      \node (t) at (7*\xsc,0)[circle, scale=4/5, draw]{$t$};

      \draw[fill=gray!20, draw] (x1) to [out=0,in=0,looseness=1.75](x2) to [out=0+180,in=+180,looseness=1.75](x1);
      \node (x1) at (4*\xsc,0)[xnode]{};
      \node (x2) at (4*\xsc,1.5)[scale=0.1]{};
      \node (Gh) at (4*\xsc,0.75)[]{$H''$};

      \draw[-]  (x1) to (w);
      \draw[-] (w) to (t);
      \draw[dashdotted] (s) to node[below,sloped]{$2n$-chain}(v);
      \draw[dashdotted] (s) to node[above]{$2n$-chain} (x1);
      \draw[-] (v) to (w);

	\node (l01) at (-0.5*\xsc,0.5)[xnode,label=90:{$u_1$}]{};
	\node (l02) at (-0.5*\xsc,-0.5)[xnode,label=-90:{$\ell_1$}]{};
    \node (l1) at (-1*\xsc,0)[xnode,label=90:{$b_1$}]{};
	\node (l11) at (-1.5*\xsc,0.5)[xnode,label=90:{$u_2$}]{};
	\node (l12) at (-1.5*\xsc,-0.5)[xnode,label=-90:{$\ell_2$}]{};
       \node (l2) at (-2*\xsc,0)[xnode,label=90:{$b_2$}]{};
	\node (ln1) at (-3*\xsc,0)[xnode,label=90:{$b_{n-2}$}]{};   
	\node (ln11) at (-3.5*\xsc,0.5)[xnode,label=90:{$u_{n-1}$}]{};   
	\node (ln12) at (-3.5*\xsc,-0.5)[xnode,label=-90:{$\ell_{n-1}$}]{};   
	 \node (ln) at (-4*\xsc,0)[xnode,label=90:{$b_{n-1}$}]{};
	 
		\node (ld) at (-2.5*\xsc,0)[]{$\ldots$};
		\draw (s) -- (l01);
		\draw (s) -- (l02);
		\draw (l1) -- (l01);
		\draw (l1) -- (l02);
		\draw (l1) -- (l11);
		\draw (l1) -- (l12);
		\draw (l2) -- (l11);
		\draw (l2) -- (l12);
		\draw (ln1) -- (ln11);
		\draw (ln1) -- (ln12);
		\draw (ln) -- (ln11);
		\draw (ln) -- (ln12);

    \draw[decorate,decoration={brace,amplitude=10pt}] (-4*\xsc,1.6-0.5) -- (-0.5*\xsc,1.6-0.5);
    \node (txt) at (-2.15*\xsc,2.15-0.5)[]{$n-1$};

  \end{tikzpicture}
  \caption{Graph~$G'$ obtained in~\cref{constr:tmtseundirhardconstdeg}. The gray part represents the graph~$H'$.}
 \label{fig:undirtrailmtseconstdeg} 
 \end{figure}
 Let initially~$G'$ be the graph obtained from~\cref{constr:tmtseundirhard}.
 Subdivide each edge in~$H'$ and denote the resulting graph by~$H''$.
 Observe that the distance in $H''$ between any two vertices in~$V(H'') \cap V(H')$ is divisible by four.
 Next, delete all edges incident with vertex~$s$.
 Connect~$s$ with~$v$ via a $2n$-chain, and connect~$s$ with~$x$ via a~$2n$-chain.
 Connect~$s$ with~$b_1$ via two~$P_2$'s.
 Denote the two vertices on the $P_2$'s by~$\ell_1$ and~$u_1$.
 Finally, for each $i\in[n-2]$, connect~$b_i$ with~$b_{i+1}$ via two~$P_2$'s.
 For each $i\in[n-2]$, denote the two vertices on the $P_2$'s between $b_i$ and $b_{i+1}$ by~$\ell_{i+1}$ and~$u_{i+1}$.
 For an easier notation, we denote vertex~$s$ also by~$b_0$.
 \qed
\end{constr}

\begin{lemma}%
 \label{lem:tmtseundirhardconstdeg}
 Let $G$ and $G'$ as in \cref{constr:tmtseundirhardconstdeg}.
 Then $G$ admits a \hcycle{} if and only if $G'$ has $2n$ $s$-$t$~trails with at most~$2n-4$ \shared edges.
\end{lemma}
\appendixproof{lem:tmtseundirhardconstdeg}
{
  \begin{proof}
  \LD{} %
  Let~$G'$ admit~a set~$\calP$ of $2n$ $s$-$t$~trails with at most~$2n-4$ \shared edges.
  At each time step, at most two trails leave~$s$ towards~$v$ and~$x$.
  Otherwise, all the edges in at least one of the $2n$-chains connecting~$s$ with~$v$ and~$s$ with~$x$ are \shared, contradicting the fact that the trails in~$\calP$ \share at most~$2n-4$ edges.
  Note that every $s$-$t$~trail contains vertex~$s$ at the first time~step and at most once more at time~step~$4j+1$, for some~$j\in \N$ (indeed, we will show that~$j\in[n-1]$).
  This follows on the one hand from the fact that~$s$ has degree four and hence every trail can contain~$s$ at most twice, and on the other hand from the fact that for each $i\in[n-1]$, every $s$-$b_i$ path is of even length.
  
  We show that at each time step~$4j+1$, $0\leq j\leq n-1$, exactly one~$s$-$t$~trial leaves~$s$ towards~$v$ and exactly one~$s$-$t$~trial leaves~$s$ towards~$x$. 
  First, observe that $|\{b_i,u_i,\ell_i\mid i\in[n-1]\}|=3(n-1)$ and each trail can contain each vertex in $\{u_i,\ell_i\mid i\in[n-1]\}\cup \{b_{n-1}\}$ at most once (as each vertex in this set is of degree two) and each vertex~$b_i$, $i\in[n-2]$, at most twice (as they are of degree four).
  Hence, any trail starting on~$s$ and returning to~$s$ after visiting the vertices in~$B$ contains at most~$3(n-1)+(n-2)+2=4(n-1)+1$ vertices.
  It follows that every $s$-$t$~trail contains~$s$ at the first time~step and at most once more at time~step~$4j+1$ for some~$j\in[n-1]$.
  As there are $2n$ $s$-$t$~trails and at each time step at most two trails leave~$s$ towards~$v$ and~$x$, together with the pigeon hole principle it follows that exactly one~$s$-$t$~trial leaves~$s$ towards~$v$ and exactly one~$s$-$t$~trial leaves~$s$ towards~$x$ at each time step~$4j+1$, $0\leq j\leq n-1$.
  Moreover, note that each $b_i$, $i\in [n-1]$, appears in at least two $s$-$t$~trails.

  We claim that there are exactly $2n-4$ \shared edges and that every shared edge is incident with a vertex in~$\{u_i,\ell_i\mid i\in[n-1]\}$. %
  This follows from the fact that at least three trails going at the same time from~$b_i$ to~$b_{i+1}$, $0\leq i\leq n-3$, \share at least two edges.
  As four trails contain~$b_{n-2}$ (those four which leave~$s$ towards~$v$ and~$x$ at the time~steps~$4j+1$ with $j\in\{n-2,n-1\}$), it follows that at least $2(n-2)$ edges are \shared.
  Hence, no two trails share an edge after they have left $s$ for $v$ or $x$.
  
  There is a trail~$P\in\calP$ that contains the vertex~$x$ and that contains vertex~$s$ only once at the first time~step, because at each time~step, two trails leave~$s$ for~$v$ or~$x$.
  Observe that vertex~$w$ is contained in the trails containing~$v$ at the time~steps~$4j+2n+2$ for all~$j\in[n-1]\cup\{0\}$ whence edge $\{w, t\}$ is occupied at time steps~$4j + 2n +3$ for all $j \in [n -1] \cup \{0\}$%
  .
  Hence, the edge~$\{w,t\}$ is contained at time~step~$2n+2$ in a trail different to~$P$ and, thus, trail~$P$ contains at least one vertex in~$H''$.
  Furthermore, $P$ can contain~$x$ a second time only at time~steps of the form~$4j+2n+1$, $3\leq j\leq n$, because each path in~$H''$ between two vertices that correspond to vertices in~$G$ has length four.
  However, as mentioned, $\{w, t\}$ is occupied at time steps $4j+2n+3$, $j \in [n - 1]$.
  Hence, $P$ has to stay in~$H''$ for~$4n$ time steps.
  Recall that~$G$ is cubic, and hence no vertex in~$H''-\{x\}$ corresponding to a vertex in~$G$ appears more than once in any trail.
  That is, $P$~follows a cycle in $H''$ containing each vertex corresponding to vertex in~$G$ exactly once.
  It follows that~$G$ admits a \hcycle{}.
  
  \RD{}
  Let~$G$ admit a \hcycle{}~$C$.
  Let~$C'$ be~$C$, ordered such that~$x$ is the first and last vertex in~$C'$.
  Let $C''$ denote the cycle in~$H''$ following the order of the vertices in~$C'$.
  We construct $2n$ $s$-$t$ trails in~$G'$ \sharing at most $2n-4$ edges as follows.
  We denote the trails by~$P_i^x$, $0\leq i \leq n-1$, $x\in\{u,\ell\}$.
  The trails are divided into two groups according to their superscript~$x\in\{u,\ell\}$.
  For $i\geq1$, the trails~$P_i^u$ and~$P_i^\ell$ start with the sequences
  \begin{align*}
    P_i^u: &&(s,\ell_1,b_1,\ldots,\ell_{i-1},b_{i-1},u_{i},b_{i},\ell_{i},b_{i-1},\ell_{i-1},\ldots,b_1,\ell_1,s),\\
    P_i^\ell: &&(s,\ell_1,b_1,\ldots,\ell_{i-1},b_{i-1},\ell_{i},b_{i},u_{i},b_{i-1},\ell_{i-1},\ldots,b_1,\ell_1,s).
  \end{align*}
  Trails $P_0^u$ and $P_0^\ell$ do not visit any vertex in~$B$ and simply start at~$s$.
  Then, for each $i = 0, \ldots, n - 1$, trail $P_i^u$ follows the chain to~$v$, the edge to~$w$, and then to~$t$.
  For each $i = 0, \ldots, n - 1$, trail $P_i^\ell$ follows the chain to~$x$, then the cycle~$C''$ in $H''$, then the edge from~$x$ to~$w$, then to~$t$.
  Observe that trail~$P_i^x$, $x\in\{u,\ell\}$, contains $s$ at time~step one and $4i+1$.
  Hence, $P_i^u$ contains the vertex~$w$ at time~step $4i+2n+2$.
  Moreover, $P_i^\ell$ contains the vertex~$x$ at time~steps~$4i+2n+1$ and~$4i+2n+4n+1$.
  From the latter it follows that~$P_i^\ell$ contains the vertex~$w$ at time~step~$4i+2n+4n+2$.
  Altogether, the edge~$\{w,t\}$ is not \shared by any pair of trails.
  
  Next, we count the number of edges \shared between the two visits of~$s$.
  Denote by $X\subseteq E(G')$ the set~$\{\{b_i,\ell_i\}\mid 1\leq i\leq n-2\}\cup\{\{b_i,\ell_{i+1}\}\mid 0\leq i< n-2\}$.
  Observe that $|X|=n-2+n-2=2(n-2)$.
  We claim that the edges in~$X$ are the only \shared edges by the trails~$P_i^x$, $0\leq i\leq n-1$, $x\in\{u,\ell\}$.
  As~$P_{n-1}^\ell$ and $P_{n-1}^u$ contain the set~$X$ at the same time~steps, every edge in~$X$ is \shared.
  For each $i\geq 1$, the edges~$\{u_i,b_{i-1}\}$ and~$\{u_i,b_i\}$ are only contained in the trails~$P_i^x$, $x\in \{u,\ell\}$.
  Recall that~$P_i^u$ and~$P_i^\ell$ contain~$b_i$ exactly once and at the same time~step. 
  The subsequence around~$b_i$ of $P_i^\ell$ and $P_i^u$ is $(b_{i-1},\ell_i,b_i,u_i,b_{i-1})$ and $(b_{i-1},u_i,b_i,\ell_i,b_{i-1})$, respectively.
  It follows that both edges~$\{u_i,b_{i-1}\}$ and~$\{u_i,b_i\}$ appear at two different time steps in~$P_i^u$ and $P_i^\ell$.
  The same argument holds for the edges~$\{b_{n-2},\ell_{n-1}\}$ and $\{b_{n-1},\ell_{n-1}\}$ as~$P_{n-1}^u$ and~$P_{n-1}^\ell$ are the only trails containing the two edges.
  
  Altogether, it follows that~$X$ is the set of \shared edges of the $s$-$t$ trails~$P_i^x$, $x\in \{u,\ell\}$, and the claim follows.
  Finally, as $|X|=2n-4$, the statement follows.
  \ifshort{}\qed\fi{}
  \end{proof}
}
\ifshort{}
The proof of \cref{thm!tmtseundirhard} then follows from~\cref{lem:tmtseundirhard} and~\cref{lem:tmtseundirhardconstdeg} (\cref{proof:thm!tmtseundirhard}).

\fi{}
\appendixproof{thm!tmtseundirhard}
{
  \begin{proof}[\iflong{}Proof of~\cref{thm!tmtseundirhard}\fi{}]
  We provide a many-one reduction from \textsc{Planar Cubic Hamilton Circuit (PCHC)} to \tmtseAcr on undirected graph via \cref{constr:tmtseundirhard} on the one hand, and~\cref{constr:tmtseundirhardconstdeg} on the other.
  Let $(G=(V,E))$ be an instance of PCHC and let $n:=|V|$~vertices.
  
  \smallskip\noindent\emph{Via \cref{constr:tmtseundirhard}.}
  Let $(G',s,t,p,0)$ an instance of \tmtseAcr where $G'$ is obtained from~$G$ by applying \cref{constr:tmtseundirhard} and~$p=2n$.
  Note that $(G',s,t,p,0)$ can be constructed in polynomial time and by \cref{lem:tmtseundirhard}, $(G)$ is a yes-instance of PCHC if and only if $(G',s,t,p,0)$ is a yes-instance of~\tmtseAcr{}.
  However, $G'$ is not simple in general.
  Hence, replace each edge~$\{u,v\}\in E(G')$ in~$G'$ by a path of length three and identify its endpoints with~$u$ and~$v$.
  Denote by~$G''$ the obtained graph.
  Due to~\cref{lem:doublesubdiv}, $(G'',s,t,p,0)$ is a yes-instance of~\tmtseAcr{} if and only if $(G',s,t,p,0)$ is a yes-instance of~\tmtseAcr{}.
  
  The case of constant~$k>0$ works analogously as in the proof of~\cref{thm!pmtsehard}.
  
  \smallskip\noindent\emph{Via \cref{constr:tmtseundirhardconstdeg}.}
  Let $(G',s,t,p,k)$ an instance of \tmtseAcr where $G'$ is obtained from~$G$ by applying \cref{constr:tmtseundirhardconstdeg} and~$p=2n$.
  Instance~$(G',s,t,p,k)$ can be constructed in polynomial time.
  By \cref{lem:tmtseundirhardconstdeg}, $(G)$ is a yes-instance of PCHC if and only if $(G',s,t,p,k)$ is a yes-instance of~\tmtseAcr{}.
    \ifshort{}\qed\fi{}
  \end{proof}
}
As the length of each $s$-$t$ trail is upper bounded by the number of edges in the graph, we immediately obtain the following.

\begin{corollary}
  \label{cor:tsmtsehardundir}
 \tsmtseAcr{} on undirected planar graphs is \NP-complete, even if~$k\geq0$ is constant or~$\Delta\geq 5$ is constant.
\end{corollary}

\subsection{On Directed Graphs}
\label{ssec:tmtse:dir}

We know that \tmtseAcr{} and \tsmtseAcr{} are \NP-complete on undirected graphs, even if the number of \shared edges or the maximum degree is constant.
In what follows, we show that this is also the case for \tmtseAcr{} and \tsmtseAcr{} on directed graphs.

\begin{theorem}%
  \label{thm!tmtsedirhard}
 \tmtseAcr on directed planar graphs is \NP-complete, even if~$k\geq0$ is constant or~$\Deltad\geq 3$ is constant.
\end{theorem}

\appendixproof{thm!tmtsedirhard}
{
  To prove~\cref{thm!tmtsedirhard}, we reduce from the following \NP-complete~\cite{Plesnik79} problem.

  \decprob{Directed Planar 2/3-In-Out Hamiltonian Circuit (DP2/3HC)}
  {A directed, planar graph $G=(V,A)$ such that, for each $v\in V$, $\max\{\outdeg(v),\indeg(v)\}\leq 2$ and $\outdeg(v)+\indeg(v)\leq 3$.}
  {Is there a directed Hamiltonian cycle in $G$?}

  \begin{constr}
  \label{constr:tmtsedirhard}
  Let~$G=(V,A)$ be a directed, planar graph where for each vertex $v\in V$ holds $\max\{\outdeg(v),\indeg(v)\}\leq 2$ and $\outdeg(v)+\indeg(v)\leq 3$, and $n=|V|$.
  Construct a directed graph~$G'$ as follows (refer to \Cref{fig:dirtrailmtse}\iflong for an illustration of the constructed graph\fi).
  \begin{figure}[t]\centering
    \begin{tikzpicture}

    \usetikzlibrary{arrows}

      \tikzstyle{xnode}=[circle, fill, scale=1/2, draw];
      \def\xsc{1.25}
      \node (s) at (0*\xsc,0)[circle, scale=3/4, draw]{$s$};
      \node (v) at (1*\xsc,0)[xnode,label=-90:{$v$}]{};

      \node (x1) at (2.5*\xsc,0)[xnode,label=90:{$x$}]{};
      \node (x2) at (2.5*\xsc,-1.5)[scale=0.1]{};

      \node (w) at (4.5*\xsc,0)[xnode,label=-90:{$w$}]{};
      \node (t) at (6*\xsc,0)[circle, scale=3/4, draw]{$t$};

      \draw[dashdotted,->,>=triangle 45] (s) to [out=25, in=155](w);
      \draw[dashdotted,->,>=triangle 45] (s) to [out=35, in=145](w);
      \draw[dashdotted,->,>=triangle 45] (s) to [out=75, in=105](w);
      \node (txt) at (2.2*\xsc,0.8)[]{$3$-chain};
      \node (txt) at (1.8*\xsc,1.1)[]{$4$-chain};
      \node (txt) at (2.3*\xsc,1.5)[]{$\vdots$};
      \node (txt) at (1.4*\xsc,1.9)[]{$(n+2)$-chain};

      \draw[fill=gray!20, draw] (x1) to [out=0,in=0,looseness=1.75](x2) to [out=0+180,in=+180,looseness=1.75](x1);
      \node (x1) at (2.5*\xsc,0)[xnode]{};
      \node (x2) at (2.5*\xsc,-1.5)[scale=0.1]{};
      \node (Gh) at (2.5*\xsc,-0.75)[scale=1.25]{$H$};

      \draw[->,>=triangle 45]  (x1) to (w);

      \draw[->,>=triangle 45] (w) to (t);
      \draw[->,>=triangle 45] (s) to (v);
      \draw[->,>=triangle 45] (v) to (x1);

      \end{tikzpicture}
    \caption{Sketch of the graph~$G'$ obtained in~\cref{constr:tmtsedirhard}. The enclosed gray part represents the graph~$H$. Dashed lines represent directed chains.}
    \label{fig:dirtrailmtse} 
    \end{figure}
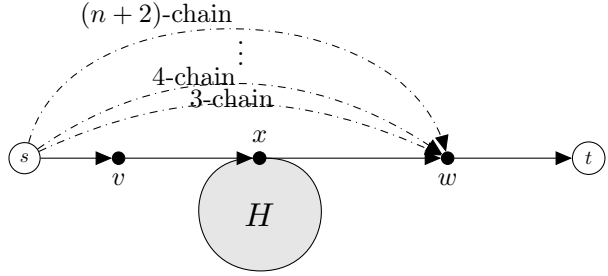
  Initially, let $G'$ be the empty graph.
  Add a copy of the graph~$G$ to $G'$ and denote the copy by~$H$.
  Add the vertex set~$\{s,t,v,w\}$ to~$G'$.
  Consider a plane embedding~$\phi(G)$ and choose a vertex~$x\in V(H)$ incident to the outer face.
  Add the arcs $(s,v)$, $(v,x)$, $(x,w)$, and~$(w,t)$ to~$G'$.
  Moreover, add $n$~chains connecting~$s$ with~$v$ of lengths~$3,4,\ldots,n+2$ respectively to~$G'$, and direct the edges from~$s$ towards~$v$.
  Note that~$G'$ is planar (see~\cref{fig:dirtrailmtse} for an embedding where~$H$ is embedded as~$\phi(H)$).
  \qed
  \end{constr}

  \begin{lemma}%
  \label{lem:tmtsedirhard}
  Let~$G$ and~$G'$ as in~\cref{constr:tmtsedirhard}.
  Then $G$ admits a \hcycle{} if and only if $G'$ admits $n+1$ $s$-$t$ trails with no \shared arc.
  \end{lemma}

  {
    \begin{proof}
    \RD{}
    Let $G$ admit a \hcycle{}~$C$. 
    We construct $n$ $s$-$t$~trails in~$G'$, each using a chain connecting $s$ with $w$, where no two use the same chain.
    By construction, the $n$ $s$-$t$~trails do not introduce any \shared arc.
    Moreover, the trails contain the arc~$(w,t)$ at every time~step in~$\{4,\ldots,n+3\}$.
    The remaining trail contains no chain, but the vertices~$v,x,w$ as well as~$C$.
    As $C$ is a \hcycle{}, the trail uses the arc~$(w,t)$ at time~step~$n+4$.
    
    \LD{}
    Let~$G'$ admit a set~$\calP$ of~$n+1$ $s$-$t$~trails with no \shared arc.
    As~$s$ has outdegree~$n+1$, in~$\calP$ $n$ trails contain a chain connecting~$s$ with~$w$, where no two contain the same chain.
    As there is no \shared arc, the remaining trail cannot use the arc $(w,t)$ before time~step~$n+3$.
    As the shortest $s$-$w$ path containing $v$ is of length three, the remaining trail has to contain $n$~arcs in the copy~$H$ of~$G$.
    As for each vertex~$v\in V(G)$ holds that~$\max\{\outdeg(v),\indeg(v)\}\leq 2$ and $\outdeg(v)+\indeg(v)\leq 3$, no vertex despite $x$ is visited twice by the trail.
    Hence, the trail restricted to the copy~$H$ of~$G$ forms an \hcycle{} in~$G$.
    \ifshort{}\qed\fi{}
    \end{proof}
  }
  We provide another construction where the obtained graph has constant maximum in- and out-degree.

  {
    \begin{constr}
    \label{constr:tmtsedirhardconstdeg}
    Let~$G=(V,A)$ be a directed, planar graph where for each vertex $v\in V$ holds $\max\{\outdeg(v),\indeg(v)\}\leq 2$ and $\outdeg(v)+\indeg(v)\leq 3$, and let $n:=|V|$.
    Construct a directed graph~$G'$ as follows (refer to \cref{fig:dirtrailmtseconstdeg} for an illustration of the constructed graph).
    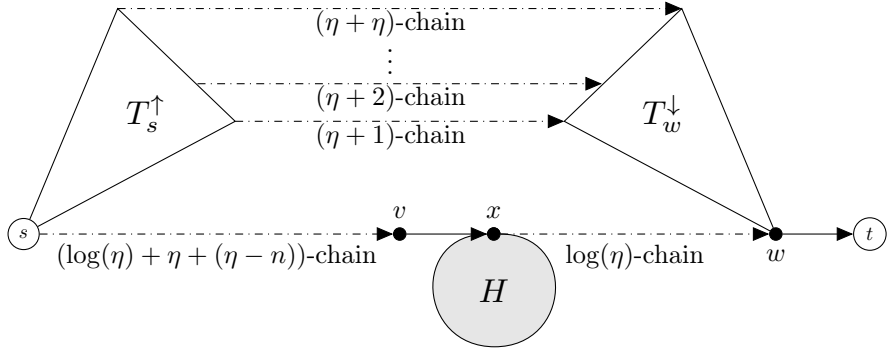
\begin{figure}[t!]
      \centering
      \begin{tikzpicture}

	\tikzstyle{xnode}=[circle, fill, scale=1/2, draw];
	\tikzstyle{xarc}=[dashdotted,->,>=triangle 45];
	\def\xsc{1.25}
	\node (s) at (-2*\xsc,0)[circle, scale=3/4, draw]{$s$};

	\node (v) at (2.0*\xsc,0)[xnode,label=90:{$v$}]{};
	\node (x1) at (3*\xsc,0)[xnode,label=90:{$x$}]{};
	\node (x2) at (3*\xsc,-1.5)[scale=0.1]{};

	\node (w) at (6*\xsc,0)[xnode,label=-90:{$w$}]{};
	\node (t) at (7*\xsc,0)[circle, scale=3/4, draw]{$t$};

	    \draw (s) -- (-1*\xsc,3) --  (0.25*\xsc,1.5) -- (s);
	    \draw (w) -- (5*\xsc,3) --  (3.75*\xsc,1.5) -- (w);
	    
	    \node (txt) at (-0.7*\xsc,1.6)[scale=1.25]{$T_s^{\uparrow}$};
	    \node (txt) at (4.8*\xsc,1.6)[scale=1.25]{$T_w^{\downarrow}$};

	\draw[xarc] (-1*\xsc,3) to (5*\xsc,3);
	\draw[xarc] (-0.15*\xsc,2) to (4.15*\xsc,2);
	\draw[xarc] (0.25*\xsc,1.5) to (3.75*\xsc,1.5);
	\node (txt) at (1.9*\xsc,1.3)[]{$(\varn+1)$-chain};
	\node (txt) at (1.9*\xsc,1.8)[]{$(\varn+2)$-chain};
	\node (txt) at (1.9*\xsc,2.4)[]{$\vdots$};
	\node (txt) at (1.9*\xsc,2.8)[]{$(\varn+\varn)$-chain};

	\draw[fill=gray!20, draw] (x1) to [out=0,in=0,looseness=1.75](x2) to [out=0+180,in=+180,looseness=1.75](x1);
	\node (x1) at (3*\xsc,0)[xnode]{};
	\node (x2) at (3*\xsc,-1.5)[scale=0.1]{};
	\node (Gh) at (3*\xsc,-0.75)[scale=1.25]{$H$};

	\draw[xarc]  (x1) to node[below]{$\log(\varn)$-chain}(w);

	\draw[->,>=triangle 45] (w) to (t);
	\draw[xarc] (s) to node[below]{$(\log(\varn)+\varn+(\varn-n))$-chain}(v);
	\draw[->,>=triangle 45] (v) to (x1);
	\end{tikzpicture}
	\caption{Sketch of the graph~$G'$ obtained in~\cref{constr:tmtsedirhardconstdeg}. 
	  The enclosed gray part represents the graph~$H$. 
	  Dashed lines represent directed chains. 
	  $T_s^{\uparrow}$ and $T_w^{\downarrow}$ refer to the complete binary (directed) trees with $\varn$ leaves and rooted at~$s$ and~$w$, respectively.}
	\label{fig:dirtrailmtseconstdeg}
    \end{figure}
    Let initially~$G'$ be the graph obtained from~\cref{constr:tmtsedirhard}.
    Remove $s$, $w$, and the directed chains connecting~$s$ with~$w$ from~$G'$.
    Let $\varn$ be the smallest power of two larger than~$n$ (note that $n\leq \varn\leq 2n-2$).
    Add a complete binary tree~$T_s$ with $\varn$~leaves to~$G'$, and denote its root by~$s$.
    Denote the leaves by $a_1,\ldots,a_\varn$, ordered by a post-order traversal on~$T_s$.
    Next add a copy of~$T_s$ to~$G'$, and denote the copy by~$T_w$ and its root by~$w$. 
    For each leaf~$a_i$ of~$T_s$, denote by $a_i'$ its copy in~$T_w$.
    Next, for each $i\in[\varn]$, connect~$a_i$ and~$a_i'$ via a~$(\varn+i)$-chain.
    Direct all edges in~$T_s$ away from~$s$ towards the leaves of~$T_s$.
    Direct all edges in~$T_w$ away from the leaves of~$T_w$ towards~$w$.
    Next, direct all chains connecting the leaves of~$T_s$ and~$T_w$ from the leaves of~$T_s$ towards the leaves of~$T_w$.
    To complete the construction of~$G'$, connect~$s$ with~$v$ via a $(\log(\varn)+\varn+(\varn-n))$-chain, and connect~$x$ with~$w$ via a $\log(\varn)$-chain.
    Note that~$T_s$ and~$T_w$ allow plane drawings, and the chains connecting the leaves can be aligned as 
    illustrated in~\cref{fig!pmtsehardconstdeg}. 
    It follows that~$G'$ allows for an plane embedding.
    \qed
    \end{constr}

    \begin{lemma}%
    \label{lem:tmtsedirhardconstdeg}
    Let~$G$ and~$G'$ as in~\cref{constr:tmtsedirhardconstdeg}.
    Then $G$ admits a \hcycle{} if and only if $G'$ admits $\varn+1$ $s$-$t$ trails with at most~$\varn-2$ \shared arcs.
    \end{lemma}

    {
      \begin{proof}
      \RD{}
      Let~$G$ admit a \hcycle{}~$C$. 
      First, we construct~$\varn$ $s$-$t$~trails in~$G'$ as follows.
      For each leaf of~$T_s$, there is a trail containing the unique path from~$s$ to the leaf in~$T_s$.
      This part introduces $\varn-2$ \shared arcs.
      Next, each trail follows the chain connecting the leaf of~$T_s$ with a leaf of~$T_w$, then the unique path from the leaf of~$T_w$ to~$w$, and finally the arc~$(w,t)$.
      Observe that the trails contain the leaves of~$T_w$ at different time steps~$\log(\varn)+\varn+i+1$, $i\in[\varn]$, and hence no \shared arc is introduced in this part.
      Moreover, the arc~$(w,t)$ appears in the time~steps~$2\log(\varn)+\varn+i+1$, $i\in[\varn]$.
      The remaining trail~$P$ contains the chain connecting~$s$ with~$v$, the edge~$\{v,x\}$, follows the cycle~$C$ in~$H$, starting and ending at vertex~$x$.
      Trail~$P$ then contains the chain connecting~$x$ with~$w$, and the arc~$(w,t)$.
      Observe that the arc~$(w,t)$ appears in~$P$ at time~step~$2\log(\varn)+2\varn+2$ (recall that~$C$ is a \hcycle{} in~$H$), and hence~$(w,t)$ is not \shared.
      
      \LD{}
      Let~$G'$ admit~$\varn+1$ $s$-$t$~trails with at most~$\varn-2$ \shared arcs.
      
      First observe that the chain connecting~$s$ with~$v$ is not contained in more than one~$s$-$t$ trail.
      Hence, at least~$\varn$ trails leave~$s$ through~$T_s$. 
      Note that no chain connecting the leaves of~$T_s$ with the leaves of~$T_w$ is contained in more than one $s$-$t$~trail.
      It follows that exactly~$\varn$ $s$-$t$~trails (denote the set by~$\calP'$) leave~$s$ via~$T_s$ and each contains a different leaf of~$T_s$.
      Herein, $\varn-2$ arcs are \shared by the trails in~$\calP'$. 
      Note that the path from a leaf of~$T_s$ to~$t$ is unique, each trail in~$\calP'$ follows the unique path to~$t$.
      The arc~$(w,t)$ appears in the trails in~$\calP'$ at time~steps~$2\log(\varn)+\varn+i+1$ for every~$i\in[\varn]$.
      
      The remaining $s$-$t$ trail~$P\not\in\calP$ contains the chain connecting~$s$ with~$v$.
      Note that arc~$(w,t)$ is not \shared, as all \shared arcs are contained in~$T_s$.
      As the shortest $s$-$t$~path via~$x$ is of length $2\log(\varn)+\varn+(\varn-n)+2 \geq 2\log(\varn)+\varn+2$, trail~$P$ has to contain a cycle~$C$ in~$H$.
      As the arc~$(w,t)$ is not \shared and appears in the trails in~$\calP'$ at the time~steps~$2\log(\varn)+\varn+i+1$, for every~$i\in[\varn]$, the cycle~$C$ must be of length~$n$.
      As for each vertex~$v\in V(G)$ holds that~$\max\{\outdeg(v),\indeg(v)\}\leq 2$ and $\outdeg(v)+\indeg(v)\leq 3$, no vertex in~$H$ despite~$x$ is visited twice by the trail~$P$.
      Hence, trail~$P$ restricted to the copy~$H$ of~$G$ forms a \hcycle{} in~$G$.
      \ifshort{}\qed\fi{}
      \end{proof}
    }
  }

  \begin{proof}[\iflong{}Proof of~\cref{thm!tmtsedirhard}\fi{}]
    We provide a many-one reduction from \textsc{DP2/3HC} to \tmtseAcr on directed graphs via \cref{constr:tmtsedirhard} on the one hand, 
    and~\cref{constr:tmtsedirhardconstdeg} on the other.
    Let $(G)$ be an instance of DP2/3HC where~$G$ consists of $n$~vertices.
  
  \smallskip\noindent\emph{Via \cref{constr:tmtsedirhard}.}   
  Let $(G',s,t,p,0)$ an instance of \tmtseAcr where $G'$ is obtained from~$G$ by applying \cref{constr:tmtsedirhard} and $p=n+1$.
  Note that $(G',s,t,p,0)$ is constructed in polynomial time and by \cref{lem:tmtsedirhard}, $(G)$ is a yes-instance of DP2/3HC if and only if $(G',s,t,p,0)$ is a yes-instance of \tmtseAcr{}.

  The case of constant~$k>0$ works analogously as in the proof of~\cref{thm!pmtsehard}.
  
  \smallskip\noindent\emph{Via \cref{constr:tmtsedirhardconstdeg}.}
  Let $(G',s,t,p,k)$ an instance of \tmtseAcr where~$G'$ is obtained from~$G$ by applying \cref{constr:tmtsedirhardconstdeg}, $p=\varn+1$, and $k=\varn-2$.
  Note that $(G',s,t,p,k)$ is constructed in polynomial time and $\max_{v\in V(G')}\{\outdeg(v)+\indeg(v)\}\leq 5$.
  By~\cref{lem:tmtsedirhardconstdeg}, $(G)$ is a yes-instance of DP2/3HC if and only 
  if $(G',s,t,p,k)$ is a yes-instance of \tmtseAcr{}.
    \ifshort{}\qed\fi{}
  \end{proof}
}
As the length of each $s$-$t$ trail is upper bounded by the number of edges in the graph, we immediately obtain the following.

\begin{corollary}
  \label{cor:tsmtseharddir}
 \tsmtseAcr{} on directed planar graphs is \NP-complete, even if~$k\geq0$ is constant or $\Deltad\geq 3$ is constant.
\end{corollary}

\section{Walk-\mtseAcr}
\label{sec:wmtse}
\appendixsection{sec:wmtse}

Regarding their computational complexity fingerprint, \pmtseAcr{} and \tmtseAcr{} are equal.
In this section, we show that \wmtseAcr{} differs in this aspect.
We prove that the problem is solvable in polynomial time on undirected graphs \iflong(\cref{ssec:wmtse:undir})\fi{} and on directed graphs if~$k\geq 0$ is constant\iflong(\cref{ssec:wmtse:dir})\fi.

\subsection{On Undirected Graphs}
\label{ssec:wmtse:undir}

\looseness=-1 On a high level, the tractability on undirected graphs is because a walk can alternate arbitrarily often between two vertices.
Hence, we can model a queue on the source vertex~$s$, where at distinct time steps the walks leave~$s$ via a shortest path towards~$t$.
\iflong{}However, if the time of staying in the queue is upper bounded, that is, if the length-restricted variant~\wsmtseAcr{} is considered, the problem becomes \NP-complete.\fi{}

\begin{theorem}%
  \label{thm!wmtseundirptime}
  \wmtseAcr{} on undirected graphs is solvable in linear time. 
\end{theorem}
\appendixproof{thm!wmtseundirptime}
{
  \begin{proof}
  Let $\I:=(G,s,t,p,k)$ be an instance of \wmtseAcr with $G$ being connected.
  Let $P$ be a shortest $s$-$t$ path in $G$.
  We assume that $p\geq 2$, since otherwise $P$ witnesses that $\I$ is a yes-instance. 
  We can assume that the length of $P$ is at least $k+1$, otherwise we can output that $\I$ is a yes-instance.
  Let~$\{s,v\}$ be the edge in~$P$ incident to the endpoint~$s$.
  We distinguish the two cases whether~$k$ is positive or~$k=0$.
  In this proof, we represent a walk as a sequence of edges.
  
  \smallskip\noindent\emph{Case $k>0$:} We can construct $p$ $s$-$t$~walks $P_1,\ldots,P_p$ \sharing at most one edge as follows.
  We set $P_1:=P$ and $P_i=(\underbrace{\{s,v\},\ldots,\{s,v\}}_{2i\text{-times}},P)$ for $i\in[p]$, that is, the $s$-$t$~walk $P_i$ alternates between $s$ and $v$ $i$ times.
  We show that the set $\calP:=\{P_1,\ldots,P_p\}$ \share exactly edge $\{s,v\}$.
  It is easy to see that $\{s,v\}$ is \shared by all of the walks.
  Let us consider an arbitrary edge~$e=\{x,y\}\neq\{s,v\}$ in $P$, which appears in $P$ at time~step $\ell>1$ (ordered from $s$ to $t$).
  By construction, $e$~appears in~$P_i$ at position~$2i+\ell$. 
  Thus, no two walks in $\calP$ contain an edge in~$P$, that is on time~step~$\ell>1$ in~$P$, at the same time~step.
  Since each walk in $\calP$ only contains edges in $P$, it follows that $\{s,v\}$ is exactly the \shared edges by all walks in~$\calP$.
  As $k\geq 1$, it follows that we can output that~$\I$ is a yes-instance.
  
  \smallskip\noindent\emph{Case $k=0$:} Let $v_1,\ldots, v_\ell$ be the neighbors of $s$, and suppose that $v_1=v$ (that is, the vertex incident to $s$ appearing in~$P$).
  If $\ell<p$, then we can immediately output that~$\I$ is a no-instance, as by the pigeon hole principle at least one edge has to appear in at least two walks at time~step one in any set of~$p$ $s$-$t$~walks in~$G$.
  If $\ell\geq p$, we construct $\ell$ $s$-$t$~walks $P_1,\ldots,P_\ell$ that do not \share any edge in~$G$ as follows.
  We set $P_1=P$, and $P_i=(\underbrace{\{s,v_i\},\ldots,\{s,v_i\}}_{2i\text{-times}},P)$ for $i\in[\ell]$, that is, the $s$-$t$~walk $P_i$ alternates between $s$ and $v_i$ $i$ times.
  Following the same argumentation as in preceding case, it follows that no edge is \shared by the constructed walks $P_1,\ldots,P_p$.
  
  In summary, if $k>0$, then we can output that $\I$ is a yes-instance.
  If $k=0$, then we first check the degree of~$s$ in linear time, and then output that~$\I$ is a yes-instance if $\deg(s)\geq p$, and that~$\I$ is a no-instance, otherwise.
  \ifshort{}\qed\fi{}
  \end{proof}
}
\iflong{}
  The situation changes for \wsmtseAcr{}, that is, when restricting the length of the walks.
\fi{}

\begin{theorem}%
 \label{thm:wsmtseundirhard}
 \wsmtseAcr{} on undirected graphs is \NP-complete and \W{2}-hard with respect to~$k+\alpha$. 
\end{theorem}
\ifshort{}
The proof is similar to the proof of~\cref{thm!dagshard} and hence deferred (\cref{proof:thm:wsmtseundirhard}).

\fi{}
\appendixproof{thm:wsmtseundirhard}
{
  Given a directed graph~$G$, we call an undirected graph~$H$ the \emph{undirected version} of~$G$ if~$H$ is obtained from~$G$ by replacing each arc with an undirected edge.

  \begin{proof}
  We give a (parameterized) many-one reduction from~\textsc{Set Cover}.
  Let $(U,\calF,\ell)$ be an instance of~\textsc{Set Cover}.
  Let $G'$ the graph obtained from applying~\cref{constr:dagshard} given~$(U,\calF,\ell)$.
  Moreover, let~$G$ be the undirected version of~$G'$.
  Let $(G,s,t,p,k,\alpha)$ be the instance of~\wsmtseAcr{}, where $p=n+m$, $k=\ell$, and~$\alpha=\ell+3$.
  Note that any shortest $s$-$t$~path in~$G$ is of length~$\alpha$, and hence every $s$-$t$~walk of length at most~$\alpha$ behaves as in the directed acyclic case.
  That is, each walk contains a vertex of $V_\calF$ at time~step~$k+2$.
  The correctness follows then analogously as in the proof of~\cref{lemma:dagshard}.
  \ifshort{}\qed\fi{}
  \end{proof}
}%
It remains open whether \wsmtseAcr{} is NP-complete when $k$ is constant.
\subsection{On Directed Graphs}
\label{ssec:wmtse:dir}

Due to~\cref{thm!dagshard,thm!dagsptime}, we know that \wmtseAcr{} is \NP-complete on directed graphs and is solvable in polynomial time on directed acyclic graphs when~$k=0$, respectively.
In this section, we prove that if~$k\geq 0$ is constant, then~\wmtseAcr{} remains tractable on directed graphs (this also holds true for \wsmtseAcr{}).
Note that for \pmtseAcr{} and \tmtseAcr{} the situation is different, as both become \NP-complete on directed graphs, even if~$k\geq 0$ is constant. 

\begin{theorem}
 \label{thm!wmtsedirxp}
 \wmtseAcr{} and \wsmtseAcr{} on directed $n$-vertex $m$-arc graphs is solvable in~$\Oh(m^{k+1}\cdot n\cdot (p\cdot n)^2)$~time and $\Oh(m^{k+1}\cdot n\cdot \alpha^2)$~time, respectively.
\end{theorem}
Our proof of~\cref{thm!wmtsedirxp} follows the same strategy as our proof of~\cref{thm!dagsxp}. That is, we try to guess the \shared arcs, make them infinite capacity in some way, and then solve the problem with zero \shared arcs via a network flow formulation in the time-expanded graph. 
The crucial difference is that here we do not have at first an upper bound on the length of the walks in the solution.

\begin{theorem}%
  \label{thm!wmtsedirk0ptime}
 If $k=0$, then \wmtseAcr{} on directed $n$-vertex $m$-arc graphs is solvable in~$\Oh(n\cdot m\cdot (p\cdot n)^2)$ time.
\end{theorem}

\begin{lemma}
  \label{lem:walkmtseShortSol}
  Every yes-instance $(G,s,t,p,k)$ of \wmtseAcr{} on directed graphs admits a solution in which the longest walk is of length at most~$p\cdot d_t$, where~$d_t=\max_{v\in V\colon\dist_G(v,t)<\infty}\dist_G(v,t)$.
\end{lemma}
\looseness=-1 Observe that~$d_t$ is well-defined on every yes-instance of \wmtseAcr{}.
Moreover, \iflong{}it holds that~\fi$d_t\leq |V(G)|$.
\iflong{}In the subsequent proof, we use the following notation:\fi{}
For two walks~$P_1=(v_1,\ldots,v_\ell)$ and $P_2=(w_1,\ldots,w_{\ell'})$ with $v_\ell=w_1$, denote by $P_1\circ P_2$ the walk $(v_1,\ldots,v_\ell,w_2,\ldots,w_{\ell'})$ obtained by the concatenation of the two walks.

\begin{proof}[\iflong{}Proof of~\cref{lem:walkmtseShortSol}\fi{}]
 Let $\calP$ be a solution to~$(G,s,t,p,k)$ with $|\calP|=p$ where the sum of the lengths of the walks in~$\calP$ is minimum among all solutions to~$(G,s,t,p,k)$.
 Suppose towards a contradiction that the longest walk~$P^*\in\calP$ is of length $|P^*|>p\cdot d_t$.
 Then, by the pigeon hole principle, there is an~$i\in[p]$ such that there is no walk in~$\calP$ of length~$\ell$ with~$(i-1)\cdot d_t<\ell\leq i\cdot d_t$.%

\looseness=-1 Let $v=P^*[(i-1)\cdot d_t+1]$, that is, $v$ is the $((i - 1) \cdot d_t+1)$th vertex on~$P^*$, and let $S$ be a shortest $v$-$t$~path.
 Observe that the length of $S$ is at most~$d_t$.
 Consider the walk~$P':=P^*[1,(i-1)\cdot d_t+1]\circ S$, that is, we concatenate the length-$((i - 1) \cdot d_t)$ initial subpath of~$P^*$ with $S$ to obtain~$P'$.
 Observe that~$(i-1)\cdot d_t<|P'|\leq i\cdot d_t$.
 If $\calP\setminus P^*\cup P'$ forms a solution to~$(G,s,t,p,k)$, then, since $|P'|<|P^*|$, $\calP\setminus P^*\cup P'$ is a solution of smaller sum of the lengths of the walks, contradicting the choice of~$\calP$.
 Otherwise, $P'$ introduce additional shared arcs and let $A'\subseteq A(G)$ denote the corresponding set.
 Observe that $A'$ is a subset of the arcs of~$S$.
 Let $a=(x,y)\in A'$ be the \shared arc such that $\dist_S(y,t)$ is minimum among all shared arcs in~$A'$, and let $P'[j]=y$.
 Let $P\in\calP$ be a walk \sharing the arc with $P'$.
 Then $P'':=P[1,j]\circ P'[j+1,|P'|]$ is a walk of shorter length than~$P$.%
 Moreover, $\calP\setminus P\cup P''$ is a solution to~$(G,s,t,p,k)$.
 As $|P''|<|P|$, $\calP\setminus P\cup P''$ is a solution of smaller sum of the lengths of the walks, contradicting the choice of~$\calP$.
 As either case yields a contradiction, it follows that $|P^*|\leq p\cdot d_t$.
 \ifshort{}\qed\fi{}
\end{proof}
\iflong{}
The subsequent proof of \cref{thm!wmtsedirk0ptime} relies on time-expanded graphs.
Due to \cref{lem:walkmtseShortSol}, we know that the time-horizon is bounded polynomially in the input size.
\fi{}%
\appendixproof{thm!wmtsedirk0ptime}
{
  \begin{proof}[\iflong{}Proof of~\cref{thm!wmtsedirk0ptime}\fi{}]
  Let $(G,s,t,p,0)$ be an instance of~\wmtseAcr{} where~$G=(V,A)$ is an directed graph.
  We first compute $d_t$ in linear time.
  Let $\tau:=p\cdot d_t$.
  Next, we compute the $\tau$-time-expanded (directed) graph~$H=(V',A')$ of~$G$ with $p$ additional arcs $(t^{i-1},t^i)$ for each $i\in [\tau]$. 
  We compute in $\Oh(\tau^2\cdot (|V|\cdot|A|)\subseteq \Oh(p^2\cdot (|V|^3\cdot|A|)$~time the value of a maximum $s^0$-$t^\tau$~flow in $H$.
  Due to~\cref{lem:walkmtseShortSol} together with~\cref{lem:flows}, the theorem follows.
  \ifshort{}\qed\fi{}
  \end{proof}
}%
Restricting to $\alpha$-time-expanded graphs yields the following.

\begin{corollary}
 \label{cor:wsmtsedirpoly}
 If $k=0$, \wsmtseAcr{} on directed $n$-vertex $m$-arc graphs is solvable in $\Oh(n\cdot m\cdot \alpha^2)$ time.
\end{corollary}

\begin{proof}[\iflong{}Proof of~\cref{thm!wmtsedirxp}\fi{}]
  Let $(G=(V,E),s,t,p,k)$ be an instance of~\wmtseAcr{} with~$G$ being a directed graph.
  For each $k$-sized subset~$K\subseteq A$ of arcs in~$G$, we decide the instance~$(G(K,p),s,t,p,0)$.
  The statement for~\wmtseAcr{} then follows from~\cref{lem:reductiontozero} and \cref{thm!wmtsedirk0ptime}.
  The running time of the algorithm is in~$\Oh(|A|^k\cdot p^2\cdot (|V|^3\cdot|A|))$.
  The statement for~\wsmtseAcr{} then follows from~\cref{lem:reductiontozero} and \cref{cor:wsmtsedirpoly}.
  \ifshort{}\qed\fi{}
\end{proof}
\section{Conclusion and Outlook}
\label{sec:concl}

Some of our results can be seen as a parameterized complexity study of \mtseAcr focusing on the number~$k$ of \shared edges. It is interesting to study the problem with respect to other parameters.
Herein, the first natural parameterization is the number of routes.
Recall that the \textsc{Minimum Shared Edges} problem is fixed-parameter tractable with respect to the number of path~\cite{FluschnikKNS15}.
A second parameterization we consider as interesting is the combined parameter maximum degree plus~$k$.
In our \NP-completeness results for \pmtseAcr and \tmtseAcr it seemed difficult to achieve constant $k$ and maximum degree at the same time.

Another research direction is to further investigate on which graph classes \pmtseAcr and \tmtseAcr become tractable.
We proved that that both problems remain \NP-complete even on planar graphs.
Do both \pmtseAcr and \tmtseAcr remain \NP-complete on graphs of bounded treewidth?
Recall that the Minimum Shared Edges problem is tractable on this graph class~\cite{YeLLZ13,AokiHHIKZ14}.

\looseness=-1 Finally, we proved that on undirected graphs, \wmtseAcr is solvable in polynomial time while \wsmtseAcr is \NP-complete. 
However, we left open whether \wsmtseAcr on undirected graphs is \NP-complete or polynomial-time solvable when~$k$ is constant.

\bibliographystyle{plain}
\bibliography{mtse-arxiv}

\newcommand{\noopsort}[1]{}
\begin{thebibliography}{10}

\bibitem{AokiHHIKZ14}
Yusuke Aoki, Bjarni~V. Halld{\'{o}}rsson, Magn{\'{u}}s~M. Halld{\'{o}}rsson,
  Takehiro Ito, Christian Konrad, and Xiao Zhou.
\newblock The minimum vulnerability problem on graphs.
\newblock In {\em Proc.~8th International Conference on Combinatorial
  Optimization and Applications ({COCOA} '14)}, volume 8881 of {\em LNCS},
  pages 299--313. Springer, 2014.

\bibitem{AssadiENYZ14}
Sepehr Assadi, Ehsan Emamjomeh{-}Zadeh, Ashkan Norouzi{-}Fard, Sadra Yazdanbod,
  and Hamid Zarrabi{-}Zadeh.
\newblock The minimum vulnerability problem.
\newblock {\em Algorithmica}, 70(4):718--731, 2014.

\bibitem{Berman96}
Kenneth~A. Berman.
\newblock Vulnerability of scheduled networks and a generalization of
  {M}enger's theorem.
\newblock {\em Networks}, 28(3):125--134, 1996.

\bibitem{DowneyF13}
Rodney~G. Downey and Michael~R. Fellows.
\newblock {\em Fundamentals of Parameterized Complexity}.
\newblock Texts in Computer Science. Springer, 2013.

\bibitem{Flu15}
Till Fluschnik.
\newblock The parameterized complexity of finding paths with shared edges.
\newblock Master thesis, In\-sti\-tut f\"ur Softwaretechnik und Theoretische
  Informatik, TU~Berlin, 2015.

\bibitem{FluschnikKNS15}
Till Fluschnik, Stefan Kratsch, Rolf Niedermeier, and Manuel Sorge.
\newblock The parameterized complexity of the minimum shared edges problem.
\newblock In {\em Proceedings of the 35th {IARCS} Annual Conference on
  Foundation of Software Technology and Theoretical Computer Science
  ({FSTTCS}'15)}, volume~45 of {\em LIPIcs}, pages 448--462. Schloss Dagstuhl -
  Leibniz-Zentrum fuer Informatik, 2015.

\bibitem{FF62}
Lester~Randolph Ford and D.~R. Fulkerson.
\newblock {\em Flows in {{Networks}}}.
\newblock {Princeton University Press}, 1962.

\bibitem{GareyJT76}
M.~R. Garey, David~S. Johnson, and Robert~Endre Tarjan.
\newblock The planar {H}amiltonian circuit problem is {NP}-complete.
\newblock {\em {SIAM} Journal on Computing}, 5(4):704--714, 1976.

\bibitem{Karp72}
Richard~M. Karp.
\newblock Reducibility among combinatorial problems.
\newblock In {\em Proceedings of a symposium on the Complexity of Computer
  Computations}, The {IBM} Research Symposia Series, pages 85--103. Plenum
  Press, New York, 1972.

\bibitem{KempeKK02}
David Kempe, Jon~M. Kleinberg, and Amit Kumar.
\newblock Connectivity and inference problems for temporal networks.
\newblock {\em Journal of Computer and System Sciences}, 64(4):820--842, 2002.

\bibitem{KT06}
Jon~M. Kleinberg and {\'{E}}va Tardos.
\newblock {\em Algorithm Design}.
\newblock Addison-Wesley, 2006.

\bibitem{KohlerMS09}
Ekkehard K{\"{o}}hler, Rolf~H. M{\"{o}}hring, and Martin Skutella.
\newblock Traffic networks and flows over time.
\newblock In {\em Algorithmics of Large and Complex Networks - Design,
  Analysis, and Simulation}, volume 5515 of {\em LNCS}, pages 166--196.
  Springer, 2009.

\bibitem{MertziosMCS13}
George~B. Mertzios, Othon Michail, Ioannis Chatzigiannakis, and Paul~G.
  Spirakis.
\newblock Temporal {{Network Optimization Subject}} to {{Connectivity
  Constraints}}.
\newblock In {\em Automata, {{Languages}}, and {{Programming}}}, pages
  657--668. {Springer}, 2013.

\bibitem{Michail16}
Othon Michail.
\newblock An {{Introduction}} to {{Temporal Graphs}}: {{An Algorithmic
  Perspective}}.
\newblock {\em Internet Mathematics}, 12(4):239--280, 2016.

\bibitem{OmranSZ13}
Masoud~T. Omran, J{\"{o}}rg{-}R{\"{u}}diger Sack, and Hamid Zarrabi{-}Zadeh.
\newblock Finding paths with minimum shared edges.
\newblock {\em Journal of Combinatorial Optimization}, 26(4):709--722, 2013.

\bibitem{Orlin13}
James~B. Orlin.
\newblock Max flows in {$O(nm)$} time, or better.
\newblock In {\em Proceedings of the 45th ACM Symposium on Theory of Computing
  ({STOC}'13)}, pages 765--774. {ACM}, 2013.

\bibitem{Plesnik79}
J{\'{a}}n Plesn{\'{\i}}k.
\newblock The {NP}-completeness of the {H}amiltonian cycle problem in planar
  digraphs with degree bound two.
\newblock {\em Information Processing Letters}, 8(4):199--201, 1979.

\bibitem{RT15}
Hassan Rashidi and Edward Tsang.
\newblock {\em Vehicle {{Scheduling}} in {{Port Automation}}: {{Advanced
  Algorithms}} for {{Minimum Cost Flow Problems}}, {{Second Edition}}}.
\newblock {CRC Press}, 2015.

\bibitem{Skutella2009}
Martin Skutella.
\newblock An introduction to network flows over time.
\newblock In {\em Research Trends in Combinatorial Optimization: Bonn 2008},
  pages 451--482. Springer, 2009.

\bibitem{YeLLZ13}
Zhi-Qian Ye, Yi-Ming Li, Hui-Qiang Lu, and Xiao Zhou.
\newblock Finding paths with minimum shared edges in graphs with bounded
  treewidths.
\newblock In {\em Proc. Frontiers of Computer Science (FCS '13)}, pages 40--46,
  2013.

\end{thebibliography}

\end{document}